\documentclass[journal]{IEEEtran}

\usepackage{tikz}
\usepackage{standalone}
\usetikzlibrary{shapes, arrows.meta, positioning, fit, backgrounds, shadows, calc}
\definecolor{mygreen}{RGB}{213,232,212}
\definecolor{mygreenline}{RGB}{98,140,94}
\definecolor{myblue}{RGB}{218,232,252}
\definecolor{myblueline}{RGB}{108,142,191}
\definecolor{myorange}{RGB}{255,242,204}
\definecolor{myorangeline}{RGB}{214,182,86}
\definecolor{mypurple}{RGB}{225,213,231}

\usepackage[ruled,linesnumbered]{algorithm2e}
\makeatletter
\newcommand{\RemoveAlgoNumber}{\renewcommand{\fnum@algocf}{\AlCapSty{\AlCapFnt\algorithmcfname}}}
\newcommand{\RevertAlgoNumber}{\algocf@resetfnum}
\makeatother

\IEEEoverridecommandlockouts

\usepackage{cite}
\usepackage{amsmath, amsthm, amssymb, amsfonts, mathtools, xfrac, bbm, bm}
\usepackage{empheq}
\usepackage[noend]{algorithmic}
\usepackage{graphicx, float}
\usepackage{textcomp}
\usepackage{xcolor}
\usepackage{booktabs, multirow}
\usepackage[Symbol]{upgreek}
\usepackage{tikz}
\usepackage[caption=false,font=footnotesize]{subfig}

\newcommand{\tph}{%
    \text{%
        \kern0.25em% <--- Adds a gap to the left of the H (Adjust 0.1em to fit)
        \begin{tikzpicture}[baseline, line cap=round]
            \draw[line width=0.4pt] (0,0) -- (0,1.5ex);        
            \draw[line width=0.4pt] (0.5em,0) -- (0.5em,1.5ex);
            \draw[line width=0.4pt] (0,0.85ex) -- (0.5em,0.85ex); 
        \end{tikzpicture}%
    }%
}

\usepackage{geometry} 
\def\BibTeX{{\rm B\kern-.05em{\sc i\kern-.025em b}\kern-.08em
    T\kern-.1667em\lower.7ex\hbox{E}\kern-.125emX}}

\graphicspath{{./Figures/}}

\theoremstyle{remark}
\newtheorem{Remark}{Remark}
\newtheorem{Lemma}{Lemma}
\newtheorem{Theorem}{Theorem}
\newtheorem{Corollary}{Corollary}

\DeclareMathOperator*{\argmin}{arg\,min}

\newcommand{\uth}{\underline{\textnormal{th}}}
\newcommand{\tsum}{\mathop{\textstyle{\sum}}}
\newcommand{\tprod}{\mathop{\textstyle{\prod}}}

\DeclarePairedDelimiter{\parens}{(}{)}
\newcommand{\trace}{\textrm{Tr}\parens}
\newcommand{\rank}[1]{\texttt{rank}\left( #1 \right)}

\newcommand{\mean}{\mathbf{E}}
\newcommand{\smean}{\mathbbm{E}}
\newcommand{\kmul}{\mathbf{K}}
\newcommand{\skmul}{\mathbbm{K}}

\newcommand{\diag}[1]{\textnormal{diag}\!\left(#1\right)}

\newcommand{\nint}[1]{\textnormal{nint}\left(#1\right)}
\newcommand{\avgsnr}{\bar{\gamma}}

\newcommand{\ibf}[1]{{ \boldsymbol{\mathit{#1}} }}

\newcommand{\bfC}{\mathbf{C}}
\newcommand{\bfD}{\mathbf{D}}

\newcommand{\bfG}{\mathbf{G}}
\newcommand{\bfH}{\mathbf{H}}
\newcommand{\bfI}{\mathbf{I}}
\newcommand{\bfK}{\mathbf{K}}

\newcommand{\bfP}{\mathbf{P}}

\newcommand{\bfR}{\mathbf{R}}

\newcommand{\bfT}{\mathbf{T}}
\newcommand{\bfTheta}{\mathbf{\Theta}}
\newcommand{\bfSigma}{\mathbf{\Sigma}}

\newcommand{\bfDelta}{\mathbf{\Delta}}
\newcommand{\bfPsi}{\mathbf{\Psi}}
\newcommand{\bfU}{\mathbf{U}}
\newcommand{\bfV}{\mathbf{V}}
\newcommand{\bfW}{\mathbf{W}}
\newcommand{\bfX}{\mathbf{X}}
\newcommand{\bfZ}{\mathbf{Z}}
\newcommand{\bfalpha}{\boldsymbol{\alpha}}
\newcommand{\bfbeta}{\boldsymbol{\beta}}
\newcommand{\bfgamma}{\boldsymbol{\gamma}}

\newcommand{\varMu}{\scalebox{1.2}{$\mu$}}
\newcommand{\varRho}{\scalebox{1.2}{$\varrho$}}
\newcommand{\scaleSigma}{\scalebox{1.1}{$\sigma$}}
\newcommand{\varP}{\scalebox{1.1}{$p$}}

\newcommand{\kron}{{\sf Kronecker}}
\newcommand{\kld}{{\sf KLD}}
\newcommand{\mmm}{{\sf MMM}}
\newcommand{\iid}{\textnormal{i.i.d.}}

\begin{document}
% \raggedbottom
% \allowdisplaybreaks
% \setlength{\textfloatsep}{0pt}
\setlength{\abovedisplayskip}{4pt}
\setlength{\belowdisplayskip}{4pt}
% \setlength{\abovecaptionskip}{-5pt}
% \setlength{\belowcaptionskip}{-10pt}
% \setlength\intextsep{0pt}
% \setlength{\dbltextfloatsep}{5pt}

% \raggedbottom
% \allowdisplaybreaks % Allow long equation to continue between pages 
\setlength{\textfloatsep}{5pt} % Remove spacing after a figure
\setlength\intextsep{0pt}
\setlength{\dbltextfloatsep}{5pt}

% \fontdimen3\font =0pt
\title{New Approximations of Non-Separable MIMO Channels by Separable Channels for Accurate Ergodic Capacity Analysis}

\author{
    \IEEEauthorblockN{
        Thanh~Luan~Nguyen\IEEEauthorrefmark{1};
        Zygmunt~J.~Haas{\IEEEauthorrefmark{2}};
        Chadi Abou-Rjeily\IEEEauthorrefmark{3};
        and ~Georges~Kaddoum{\IEEEauthorrefmark{1}}; 
    } \\
    \IEEEauthorrefmark{1}Department of Electrical Engineering, \'{E}cole de Technologie Sup\'{e}rieure (\'{E}TS), Montr\'{e}al, QC, Canada \\
    \IEEEauthorrefmark{2}School of Electrical and Computer Engineering, Cornell University, Ithaca, NY 14853. \\
    \IEEEauthorrefmark{3}Department of Electrical and Computer Engineering, Lebanese American University (LAU), Byblos 5053, Lebanon. \\
    Emails: 
    thanh-luan.nguyen.1@ens.etsmtl.ca,
    chadi.abourjeily@lau.edu.lb,
    zhaas@cornell.edu,
    georges.kaddoum@etsmtl.ca.
}

\maketitle

\begingroup
\renewcommand{\thefootnote}{}
\footnotetext{This manuscript has been submitted for publication.}
\addtocounter{footnote}{-1}
\endgroup

\begin{abstract}
In recent years, owing  to the high accuracy in characterizing non-separable  channels prevalent  in next-generation wireless applications, the classical Weichselberger channel model has gained widespread adoption in multiple-input multiple-output (MIMO) systems.     However, its non-separable structure also introduces severe analytical complexity, leading to a lack of tractable mathematical frameworks in the literature and thus raises an urgent need for further research.
To address the aforementioned analytical complexity, we first derive the nearest separable (double-correlated Rayleigh) fading model to the Weichselberger model under the Kullback–Leibler divergence (KLD), a problem equivalent to rank-$1$ nonnegative matrix factorization under the Itakura–Saito (IS) distance criterion.
    The results of our asymptotic analysis in the high-SNR regime reveal that the KLD-enabled approximation achieves a tighter capacity estimate than the conventional Kronecker model, especially in sparse and non-regular scattering environments.
Yet, a key limitation of the KLD-enabled model is its tendency to mischaracterize the channel capacity in the low-SNR regime due to its inability to preserve total channel power.
    As a more robust alternative, we introduce a novel moment matching method (MMM) aimed at mapping the exact channel statistics to those of a Wishart distribution.
Both the KLD-enabled and MMM-enabled separable channel directly enable the use of exact closed-form expressions for the ergodic capacity. 
    Numerical results demonstrate that the MMM-enabled model consistently improves upon the capacity accuracy of the conventional Kronecker model across all SNR regimes.
\end{abstract}


\section{Introduction}
\IEEEPARstart{T}{HE} global evolution towards next-generation wireless networks has increased the development and deployment of multiple-input multiple-output (MIMO) technologies to achieve further gains in capacity and reliability over bandwidth-limited channels \cite{Zhang2018CM}, \cite{bjornson2019massive}, \cite{rappaport2019wireless}, \cite{wang2023road}.
    However, an accurate quantification of such performance gains largely depends on the ability of the underlying channel model to capture different properties of real-world propagation environments.
Driven by such requirements, in recent years, considerable efforts have been invested in the development of channel models that would balance the trade-off between complexity and accuracy for reliable MIMO system performance evaluation \cite{ozcelik2005makes}, \cite{wang2023pervasively}, \cite{demir2024spatial}, \cite{psychogios2025dual}.

    A main challenge in such accurate evaluation of MIMO system performance is the presence of spatial correlation. 
    In real-world deployments, spatial correlation is inevitable and can degrade the MIMO capacity \cite{sayeed2002deconstructing, ozcelik2005makes}. 
Zero spatial correlation requires highly idealized (and practically unrealistic) conditions, such as isotropic scattering combined with uniform linear arrays (ULAs) spaced at exact integer multiples of half the wavelength. 
    Since realistically uncorrelated environments are rarely observed, tremendous effort has been invested into characterizing correlated MIMO channels. 
    One of the most well-known frameworks in the literature is the double-correlated Rayleigh model, which introduces a Kronecker correlation structure \cite{kermoal2002stochastic, ying2014kronecker}. 
The Kronecker framework simplifies mathematical analysis by assuming that spatial correlation at the transmitter operates independently from spatial correlation at the receiver, resulting in a separable variance profile.
    
    However, the Kronecker model fails to capture physical dependence between the transmitting and receiving antenna arrays and, consequently, frequently underestimates the true channel capacity in realistic propagation environments \cite{Raghavan2010TIT}.
Another critical limitation of the Kronecker model is its inability to accurately predict capacity and diversity even for large antenna arrays \cite{ozcelik2005makes}, which makes it unsuitable for accurate performance analysis of Extremely Large Aperture Array (ELAA) technologies \cite{rappaport2019wireless}, which is an advanced wireless antenna architecture that features thousands of closely coordinated elements spanning large physical dimensions.
    To overcome the aforementioned limitations, a more generalized alternative that was recently proposed is the Weichselberger model, which introduces a non-separable correlation structure to characterize arbitrary coupling between transmit and receive antennas while encompassing the Kronecker model as a simplified case \cite{kermoal2002stochastic}, \cite{Weichselberger2006TWC}, \cite{Raghavan2010TIT}. 
As a result, the Weichselberger model provides a more general structure in representing a variety of channel scenarios. 
    Recently, this model has  become essential for evaluating emerging next-generation architectures \cite{psychogios2025dual}, such as reconfigurable intelligent surface (RIS)-aided channels \cite{ren2026multi} and holographic MIMO systems \cite{Zhang2025TIT}.

Yet, despite the broad applicability of the Weichselberger model, its underlying non-separable structure introduces significant analytical complexities, as it often relies on analytical upper bounds \cite{gao2009statistical} or large-system asymptotic approximations \cite{wen2011ergodic, Raghavan2010TIT, Zhang2025TIT}. 
    Such methods can fall short when exact performance evaluation is required for systems with finite-dimensional antenna arrays.
    A typical example is the fully-connected beyond-diagonal RIS architecture \cite{Li2026tutotial}, where hardware-induced attenuation scales exponentially as a result of mutual coupling between elements and limits the maximum number of deployable elements.
This induces a critical need for tractable mathematical approximations capable of preserving the underlying statistical properties of the Weichselberger model across both finite-dimensional as well as large-scale configurations.

\subsection{Related Works}

The characterization of Ergodic capacity (EC) and outage probability (OP) frequently depends on whether the channel matrices are centered, separable, and analyzed in the asymptotic or finite (non-asymptotic) regimes. 
    For {centered and non-separable channels}, in \cite{gao2009statistical}, the authors derived a non-asymptotic, closed-form tight upper bound for the EC using the matrix permanent (i.e., perm()) of the eigenmode channel coupling matrix. 
In the large-system limit for {non-centered and non-separable channels}, in \cite{wen2011ergodic} and \cite{wen2011sum}, the authors evaluated the asymptotic sum-rate capacity of multiuser MIMO uplink channels as the number of transmit and receive antennas approaches infinity. 
    Taking a finite-system approach, in \cite{wang2023pervasively}, the authors proposed a pervasively correlated channel model to evaluate the capacity distributions via numerical sampling. 
In another study that relied on the Random Matrix Theory (RMT) \cite{Zhang2025TIT}, the authors established a central limit theorem (CLT) to provide closed-form, asymptotic expressions for the mean and variance of the MI in channels containing both line-of-sight (LOS) and non-line-of-sight (NLOS) components.

In one of the earliest studies \cite{kollo1995approximating}, the authors established a formal framework for approximating unknown multivariate probability density functions (PDFs) of a symmetric random matrix by using Edgeworth series expansion around the centered
Wishart base.
    Due to their reliance on Kronecker-structured matrix moments, the dimensionality of higher-order moments was found to explode exponentially, thus complicating their moment-matching framework for heavily skewed distributions. 
By contrast, subsequent studies \cite{hillier2021moments} and \cite{maiwald2000calculation} derived the true matrix-valued moments of real-valued Wishart and complex-valued Wishart random matrices, respectively, which preserve the original matrix dimensions.
    In the closest relevant work to this paper, \cite{Pivaro2017TVT} proposed a Wishart matrix approximation to derive non-asymptotic, closed-form expressions for the ergodic sum-rate capacity of MIMO channels via exact marginal eigenvalue distributions. 
Although the aforementioned study was strictly limitted to the sum of isotropic complex Wishart matrices, it provided a fresh moment-matching perspective that motivates the need for a generalized approximation framework for Wishart matrices with arbitrary covariance structures. 
    Should such an approximation framework be established, analytical results for separable channels could be directly adapted to accurately explore the performance of future wireless communication systems.
Indeed, in the literature, a wealth of exact and asymptotic solutions for these separable models was already proposed.

For {centered and separable channels}, in \cite{kiessling2004exact}, the authors provided a strictly non-asymptotic analysis, characterizing the exact EC as a sum of determinants. 
    Using asymptotic RMT, in \cite{hachem2008new} and \cite{bao2015asymptotic}, the authors investigated the Mutual Information (MI) distribution, establishing methods that remain accurate even when applied to non-asymptotic, finite-dimensional systems.
For {non-centered and separable channels}, such as the correlated Rician channels, in \cite{dumont2010capacity}, the authors optimized capacity-achieving covariance matrices by relying on an asymptotic approximation of the average MI. 
    Through a PDF-enabled framework, another relevant work \cite{maaref2007joint} derived exact, non-asymptotic eigenvalue distributions for non-central Wishart matrices to compute the mean, variance, and higher-order statistics of the random channel capacity.
    
\subsection{Contributions}

The main analytical complexity of non-separable Weichselberger channels originates from the power coupling matrix, denoted as $\bfP$.
    To resolve the aforementioned analytical complexity, in our present study, we propose two distinct approximations that specifically target the stochastic NLOS component. 
By focusing our rank-$1$ factorizations solely to $\bfP$, the deterministic LOS component remains perfectly exact and also prevents it from masking the true approximation errors during performance evaluation. 
    Building upon this strategy, the main contributions of this paper can be summarized as follows:
\begin{itemize}
    \item We formulate a computationally efficient Kullback–Leibler Divergence (KLD)-enabled {rank-$1$} approximation algorithm for non-separable power coupling matrices. The proposed algorithm demonstrates rapid convergence and accurately tracks the exact Weichselberger eigenvalue distribution in standard scattering environments.
    \item We establish an analytically tractable moment matching method (MMM) procedure that maps the exact channel statistics onto an equivalent Wishart matrix by solving a Stieltjes moment problem. The proposed framework is enabled by deriving novel recursive expressions for the matrix-valued moments/cumulants of the sum of complex Wishart matrices with diagonal covariance.
    \item We also derive exact, closed-form analytical expressions for the EC of the proposed KLD-enabled and MMM-enabled rank-$1$ model. In addition, we provide an asymptotic capacity gap analysis in the high-signal-to-noise ratio (SNR) extreme to formally establish that the proposed KLD-enabled model achieves a tighter capacity gap as compared to the conventional Kronecker baseline.
    \item Through extensive numerical validation, we demonstrate that the MMM-enabled model serves as a robust alternative to the KLD-enabled model that accurately approximates the dominant eigenvalue moments. 
    Specifically, it consistently improves upon the capacity accuracy of the traditional Kronecker baseline across all signal-to-noise ratio (SNR) regimes, while providing highly reliable capacity estimates even for highly structured power coupling.
\end{itemize}

\textit{Notation}: We use the $\diag{\cdot}$ operator that puts the elements of a $d$-dimensional vector onto the main diagonal of an
$d\times d$ diagonal matrix, and the $\trace{\cdot}$ operator that denotes the trace of a matrix. 
    Superscripts $(\cdot)^{\top}$ and $(\cdot)^{\tph}$ denote the matrix transpose and Hermitian matrix transpose, respectively. 
Symbols $\odot$, $(\cdot)^{\circ\frac{1}{2}}$, $(\cdot)^{\circ-1}$ denote the element-wise product, square root, and inverse operators, respectively.
    Boldface letters indicate matrices, such that $\mean[\cdot]$ and $\kmul_n[\cdot]$ return the matrix-valued expectation and $n^{\uth}$ order cumulant, as opposed to their scalar counterparts $\smean[\cdot]$ and $\skmul_n[\cdot]$. % The distinction is purely for readability.

\section{Weichselberger MIMO Channel Model}

The general non-centered (Rician) Weichselberger MIMO channel matrix ${\bfH} \in \mathbbm{C}^{R\times T}$ is modeled as the superposition of a deterministic LOS component and a stochastic NLOS component, where $R$ and $T$ are the numbers of receiving and transmitting antennas, respectively, given by the following~\cite{Zhang2025TIT}:
\begin{align}
\bfH_{\textnormal{tot}} 
    &=   \bar{\bfH} + {\bfH} \\
    &=  \bar{\bfH} + \bfU_{\bfR} \left[ 
        \bfP^{\circ\frac{1}{2}} \odot \bfH_{\iid} 
    \right] \bfU_{\bfT}^{\tph},
\label{eq:Rician_Weichselberger_H}
\end{align}
where 
    $\bar{\bfH} \in \mathbbm{R}^{R\times T}$ and 
    ${\bfH} \in \mathbbm{C}^{R\times T}$ are the LOS and NLOS components, respectively,
    $\bfH_{\iid} \in \mathbbm{C}^{R\times T}$ represents the channel matrix of the spatially uncorrelated Rayleigh fading, whose entries are independent and identically distributed (i.i.d.) zero-mean complex Gaussian with unit variance, 
    $\bfU_{\bfR} \in \mathbbm{C}^{R\times R}$ and 
    $\bfU_{\bfT} \in \mathbbm{C}^{T\times T}$ are the deterministic matrices representing the unitary eigenbases of the receiver and transmitter, respectively, with the columns $\mathbf{u}_{\bfR, r}$ and $\mathbf{u}_{\bfT, t}$ denoting the corresponding $r^{\uth}$ receive and $t^{\uth}$ transmit spatial eigenmodes \cite{ozcelik2005makes, Weichselberger2006TWC}. 
The eigenbases are defined via the eigendecomposition of the one-sided spatial correlation matrices $\bfR \in \mathbbm{C}^{R\times R}$ and $\bfT \in \mathbbm{C}^{T\times T}$ as follows:
\begin{align}
\bfR &\triangleq \mean\left[ \bfH \bfH^{\tph} \right]
    = \bfU_{\bfR} \bfSigma_{\bfR} \bfU_{\bfR}^{\tph}, \\
\bfT &\triangleq \mean\left[ \bfH^{\tph} \bfH \right]
    = \bfU_{\bfT} \bfSigma_{\bfT} \bfU_{\bfT}^{\tph},
\end{align}
where matrices  
    $\bfSigma_{\bfR} = \diag{\boldsymbol{\sigma}_{\bfR}}$ and
    $\bfSigma_{\bfT} = \diag{\boldsymbol{\sigma}_{\bfT}}$ consist of the corresponding non-zero, real-valued eigenvalues with 
    $\boldsymbol{\sigma}_{\bfR} \triangleq \left[ \sigma_{\bfR, 1}, \sigma_{\bfR, 2}, \dots, \sigma_{\bfR, R} \right]^{\top}$ and 
    $\boldsymbol{\sigma}_{\bfT} \triangleq \left[ \sigma_{\bfT, 1}, \sigma_{\bfT, 2}, \dots, \sigma_{\bfT, T} \right]^{\top}$.
    
Moreover, matrix $\bfP = \mean\left[ \left| \bfU_{\bfR}^{\tph} \bfH \bfU_{\bfT} \right|^2 \right] \in \mathbbm{R}^{R\times T}$ denotes the power coupling matrix, which characterizes the average power transfer between the transmit and receive eigenmodes,
    where elements of $\bfP$ are related to the one-sided eigenvalues as shown below:
\begin{align}
\sum_{r=1}^{R}{ P_{r, t} } = \sigma_{\bfT, t}, 
\quad 
\sum_{t=1}^{T}{ P_{r, t} } = \sigma_{\bfR, r},
\label{eq:5}
\end{align}
where $\sigma_{\bfT, 1} \ge \sigma_{\bfT, 2} \!\ge \cdots \!\ge \sigma_{\bfT, T}$ and 
    $\sigma_{\bfR, 1} \!\ge \sigma_{\bfR, 2} \!\ge \cdots \!\ge \sigma_{\bfR, R}$ which are normalized as follows:
\begin{align}
\sum_{r=1}^{R} \sum_{t=1}^{T}{ P_{r, t} } = P_{\bfH},
\label{eq:6}
\end{align}
where $P_{\bfH} \triangleq \smean\left\{ \trace{\bfH \bfH^{\tph}} \right\} = \smean\left\{ \trace{\bfH^{\tph} \bfH} \right\}$ denotes the total mean power of the NLOS components, often set to $P_{\bfH} = RT$.

\begin{Remark}
In the traditional Kronecker model, the power coupling matrix is decomposed as
    $\bfP = \bfP_{\kron} = \frac{1}{P_{\bfH}} \boldsymbol{\sigma}_{\bfR} \boldsymbol{\sigma}_{\bfT}^{\top}$.
Consequently, the channel model is separable and can be formulated as shown below:
\begin{align}
{\bfH}_{\kron} = \frac{1}{\sqrt{P_{\bfH}}} \bfU_{\bfR} \bfSigma_{\bfR}^{\sfrac{1}{2}} \bfH_{\iid} \bfSigma_{\bfT}^{\sfrac{1}{2}} \bfU_{\bfT}^{\tph},
\label{eq:H_kron}
\end{align}
which is frequently referred to as the spatially double-correlated Rayleigh fading channel model \cite{bao2015asymptotic, hachem2008new}.
\end{Remark}

While the general Rician framework effectively captures environments with a dominant direct path, 
in the remainder of this paper, we will focus entirely on the classical Weichselberger framework \cite{Weichselberger2006TWC}, such that:
\begin{align}
    \bfH \triangleq \bfU_{\bfR} \left[ 
        \bfP^{\circ\frac{1}{2}} \odot \bfH_{\iid} 
    \right] \bfU_{\bfT}^{\tph}.
\label{eq:Weichselberger_H}
\end{align}

Focusing our rank-$1$ factorization above to $\bfP$ ensures the deterministic LOS component remains perfectly exact.
As a result, upon re-integrating the approximation of Eq. \eqref{eq:Weichselberger_H} into Eq. \eqref{eq:Rician_Weichselberger_H}, the resulting channel effectively collapses into a simpler and well-studied double-correlated Rician model without inducing additional errors into the direct LOS path.

In this approach, the instantaneous channel power gain is defined as follows:
\begin{subequations}
\label{eq:GH}
\begin{empheq}[left={\mathbf{G}_{\bfH} \triangleq \empheqlbrace}]{align}
    \mathbf{H}^{\tph} \mathbf{H}, & \quad R \ge T, \label{eq:GHa} \\
    \mathbf{H} \mathbf{H}^{\tph}, & \quad R < T. \label{eq:GHb}
\end{empheq}
\end{subequations}

Let us denote 
    $\mathbf{r}_{r} \in \mathbbm{R}^{T}$ and 
    $\mathbf{c}_{t} \in \mathbbm{R}^{R}$ as the $r^{\uth}$ row and $t^{\uth}$ column of $\bfP$, respectively. Then, 
the stochastic representations of $\bfG_{\bfH}$ in Eq. \eqref{eq:GHa} and Eq. \eqref{eq:GHb} are rewritten by the following:
\begin{align}
{\bfG}_{\bfH}
    &=
    \! \bfU_{\bfT} \!
    \sum_{r=1}^{R}
    \diag{\mathbf{r}_{r}}^{\frac{1}{2}} \!
    \bfH_{\iid, (r, :)}^{\tph} \bfH_{\iid, (r, :)}
    \diag{\mathbf{r}_{r}}^{\frac{1}{2}} \!
    \bfU_{\bfT}^{\tph}, \! \label{eq:GH_Stochastic_ROW}
\end{align}
for $R \ge T$ and:
\begin{align}
{\bfG}_{\bfH}
    &=
    \! \bfU_{\bfR} \!
    \sum_{t=1}^{T}
    \diag{\mathbf{c}_{t}}^{\frac{1}{2}} \!
    \bfH_{\iid, (:, t)} \bfH_{\iid, (:, t)}^{\tph}
    \diag{\mathbf{c}_{t}}^{\frac{1}{2}} \!
    \bfU_{\bfR}^{\tph}, \! \label{eq:GH_Stochastic_COL}
\end{align}
for $R < T$, where $\bfH_{\iid, (r, :)}$ and $\bfH_{\iid, (:, t)}$ denote the $r^{\uth}$ row and the $t^{\uth}$ column of $\bfH_{\iid}$, respectively.

\textbf{\textit{Permutation and zero-row/column truncation:}}
    As implied by Eq. \eqref{eq:GH_Stochastic_ROW}, the distribution of $\bfG_{\bfH}$ in the case $R \ge T$ is preserved by the removal of zero-valued rows from $\bfP$ and is invariant to the row permutations. 
    Similarly, in the case $R < T$, the distribution of $\bfG_{\bfH}$ is preserved by the removal of zero-valued columns from $\bfP$ and is invariant to the column permutations.
Note that simultaneously performing row and column permutations affects the distribution of the channel power gain.

In light of the above observations, and without loss of generality, we henceforth assume the following:
\begin{itemize}
\item {\textit{The coupling matrix $\bfP$ contains no all-zero rows or columns}}, as they can be ignored 
    without affecting the spectral distribution of $\bfG_{\bfH}$.
\item {\textit{The system operates in the regime where $R \ge T$}}, since the exact same analytical framework is applicable for the regime where $R \le T$
    by swapping the roles of the transmitter and receiver. 
\end{itemize}

\section{KLD-Enabled Rank-$1$ Factorization}

The nearest double-correlated Rayleigh fading model to the Weichselberger model under the Kullback–Leibler divergence (KLD) criterion is derived as follows:
\begin{align}
{\bfP}_{\kld}
    =  \argmin_{\bfalpha, \bfbeta} 
        \quad & {D}_{\rm KL} \left( {\bfH} \parallel \widetilde{\bfH} \right) \\
    &= \sum_{r=1}^{R} \sum_{t=1}^{T} 
        {D}_{\rm KL}\left( {H}_{r, t} \parallel \widetilde{H}_{r,t} \right), \label{eq:10}
\end{align}
where $\bfalpha \in \mathbbm{R}^{R}_{+}$,
    $\bfbeta \in \mathbbm{R}^{T}_{+}$, 
    and 
\begin{align}
    \widetilde{\bfH} \triangleq \bfU_{\bfR} \left[ \left( \bfalpha \bfbeta^{\top} \right)^{\circ\frac{1}{2}} \odot \widetilde{\bfH}_{\iid} \right] \bfU_{\bfT}^{\tph}
\end{align}
with entries of $\widetilde{\bfH}_{\iid} \in \mathbbm{C}^{R\times T}$ following i.i.d. zero-mean complex Gaussian with unit variance. 
    The objective function is derived as follows \cite{chen2025divergence}:
\begin{align}
J\left( \bfalpha, \bfbeta \right)
    =   - RT + \sum_{r=1}^{R} \sum_{t=1}^{T}
    \frac{ P_{r, t} }{ \alpha_r \beta_t } - \ln \left( \frac{ P_{r, t} }{ \alpha_r \beta_t } \right).
\label{eq:12}
\end{align}

Therefore, the problem becomes the rank-$1$ nonnegative matrix factorization (NMF) of the matrix $\bfP$ under the Itakura-Saito (IS) distance criterion \cite{fevotte2009nonnegative}.
    It is noted that the objective function is jointly convex under he logarithm transformation $\mathbf{x} = \ln\bfalpha$ and $\mathbf{y} = \ln\bfbeta$.
Using this property, we can generate the following sequences $\{ \bfalpha^{(\kappa)} \}$ and $\{ \bfbeta^{(\kappa)} \}$ iteratively. At each $\kappa^{\uth}$ iteration,
\begin{itemize}
    \item $\bfalpha^{(\kappa)}$ is found by solving ${\frac{\partial}{\partial \alpha_r} J\left( \bfalpha, \bfbeta^{(\kappa-1)} \right) = 0}$ for $\bfalpha$, which guarantees  
    ${J\left( \bfalpha^{(\kappa)}, \bfbeta^{(\kappa-1)} \right) \le J\left( \bfalpha^{(\kappa-1)}, \bfbeta^{(\kappa-1)} \right)}$;
    \item $\bfbeta^{(\kappa)}$ is found by solving $\frac{\partial}{\partial \beta_t} J\left( \bfalpha^{(\kappa)}, \bfbeta \right) = 0$ for $\bfbeta$, which guarantees 
    $J\left( \bfalpha^{(\kappa)}, \bfbeta^{(\kappa)} \right) \le J\left( \bfalpha^{(\kappa)}, \bfbeta^{(\kappa-1)} \right)$.
\end{itemize}

Consequently, the above procedure guarantees a monotonic decrease of the objective function at each iteration, such that:
\begin{align}
    J\left( \bfalpha^{(\kappa)}, \bfbeta^{(\kappa)} \right) 
    \le J\left( \bfalpha^{(\kappa-1)}, \bfbeta^{(\kappa-1)} \right).
\end{align}

Of note, solving ${\frac{\partial J\left( \bfalpha, \bfbeta \right)}{\partial \alpha_r} \!= 0}$ and 
    $\frac{\partial J\left( \bfalpha, \bfbeta \right)}{\partial \beta_t} \!= 0$ yields the following vector forms:
%{subequations}
\begin{align}
\bfalpha = \frac{1}{T} {\bfP} \bfbeta^{\circ-1}, \quad
\bfbeta = \frac{1}{R} {\bfP}^{\top} \bfalpha^{\circ-1},
\label{eq:14}
\end{align}
respectively.

The procedure to solve the problem in Eq. \eqref{eq:10} can be summarized in Algorithm \ref{alg:optimal_rank1_nmf}.

\begin{algorithm}[!h]
\caption{Optimal Rank-$1$ Approximation with IS-NMF Divergence\label{alg:optimal_rank1_nmf}}
\begin{algorithmic}[1]
\STATE \textbf{Initialization}: Set $\kappa = 0$ and $J^{(0)} = \infty$. Choose $\bfalpha^{(0)} = \mathbf{1}_R$, and initialize $\bfbeta^{(0)} = \frac{1}{R} \mathbf{P}^{\top} \left( \bfalpha^{(0)} \right)^{\circ-1}$;
\REPEAT
    \STATE Set $\kappa \gets \kappa + 1$;
    \STATE Update $\bfalpha^{(\kappa)} \gets \frac{1}{T} \mathbf{P} \left( \bfbeta^{(\kappa-1)} \right)^{\circ-1}$;
    \STATE Update $\bfbeta^{(\kappa)} \gets \frac{1}{R} \mathbf{P}^{\top} \left( \bfalpha^{(\kappa)} \right)^{\circ-1}$;
    \STATE Update $J^{(\kappa)} \gets J\left( \bfalpha^{(\kappa)}, \bfbeta^{(\kappa)} \right)$;
\UNTIL{$\left| J^{(\kappa-1)} - J^{(\kappa)} \right| < \epsilon$}
\STATE Compute ${\mathbf{P}}_{\kld} := \bfalpha^{(\kappa)} \left( \bfbeta^{(\kappa)} \right)^{\top}$;
\end{algorithmic}
\end{algorithm}

    Importantly, the computational complexity of the above procedure is $O(RT)$ per iteration.
Since $J\left( \bfalpha, \bfbeta \right) \ge 0$ and it monotonically decreases with each updated $\bfalpha^{(\kappa)}$ and $\bfbeta ^{(\kappa)}$, the algorithm is guaranteed to converge to a stationary point.
    Moreover, the objective function becomes jointly convex the logarithm transformation $\mathbf{x} = \ln\bfalpha$ and $\mathbf{y} = \ln\bfbeta$. 
Therefore, every stationary point is globally optimal, implying that Algorithm \ref{alg:optimal_rank1_nmf} converges to the global optimal rank-$1$ approximation
${\mathbf{P}}_{\kld} = \bfalpha^{\star} \left( \bfbeta^{\star} \right)^{\top}$ under KLD criterion with equality $J^{\star} = 0$ if and only if $\bfP$ is already rank-$1$.
    
\section{MMM-Enabled Rank-$1$ Factorization}
\label{sec:rank_1_Approx_via_mmm}

In this section, we adopt the MMM to map the exact channel statistics onto a simplified model via their moments and/or cumulants \cite{golub2009matrices}. 
    However, a direct application of MMM appears to be ineffective, as forcing a non-separable channel into a separable channel limits the order of moments that can be matched. 
As a workaround, we first approximate normalized alternatives of $\bfH$ and then recover the desired channel approximation.
    The normalization aims to remove dominant row and column power imbalances from $\bfP$, yielding a variance profile that is closer to an isotropic Wishart structure and, therefore, more amenable to moment/cumulant matching.
The proposed approximation framework consists of the following three main steps:

\begin{itemize}
    \item[1)] {\it Channel normalization.} We first apply scaling matrices to yield a normalized power coupling matrix:
    \begin{align}
        {\bfP}_{\sf Normalized} 
            =   \diag{\boldsymbol{s}_{\bfR}}^{-1} \bfP \diag{\boldsymbol{s}_{\bfT}}^{-1}. 
        \label{eq:P_normalized}
    \end{align}
    where
        $\diag{\boldsymbol{s}_{\bfT}} \in \mathbbm{R}^{T\times T}$ and 
        $\diag{\boldsymbol{s}_{\bfR}} \in \mathbbm{R}^{R\times R}$ are~predefined diagonal transmit and receive scaling matrices.
    The channel gain corresponding to the above variance profile can be written as $\bfU_{\bfT} \bfW \bfU_{\bfT}^{\tph}$ with:
    \begin{align}
    \bfW &\triangleq 
        \sum_{r=1}^{R}  
        \bfSigma_r^{\frac{1}{2}} 
        \bfH_{\iid, (r, :)} \bfH_{\iid, (r, :)}^{\tph}
        \bfSigma_r^{\frac{1}{2}} \\
        &\triangleq \sum_{r=1}^{R}{ \bfW_r }, \notag
    \label{eq:bfW_def}
    \end{align}
    with $\bfSigma_{r} \!\triangleq \frac{1}{[\boldsymbol{s}_{\bfR}]_{r}} \diag{\mathbf{r}_{r}} \diag{\boldsymbol{s}_{\bfT}}^{-1}$.
    Each $\bfW_r \!\in \mathbbm{C}^{T\times T}$ component follows a singular complex Wishart distribution with a single degree of freedom and a diagonal covariance matrix $\bfSigma_{r}$.
    \item[2)] {\it Wishart approximation.} We apply the MMM to match the eigenvalue moments of $\bfW$ to those of a more tractable, Wishart-distributed matrix
    $\widetilde{\bfW} \in \mathbbm{C}^{T\times T}$, defined as follows:
    \begin{align}
        \widetilde{\bfW} = \widetilde{\bfH}_{\iid}^{\tph} {\bfSigma}_{\bfgamma} \widetilde{\bfH}_{\iid},
    \end{align}
    where ${\bfSigma}_{\bfgamma} \in \mathbbm{R}^{R\times R}$ is a strictly diagonal matrix whose diagonal elements are given by vector $\bfgamma \in \mathbbm{R}^{R}$.
    \item[3)] {\it Channel recovery.} 
    Once parameter vector $\bfgamma$ is obtained, which is equivalent to obtaining the rank-$1$ factorization for~the matrix 
    ${\bfP}_{\sf Normalized}$, the MMM-enabled rank-$1$ factorization of $\bfP$ is formulated as follows:
    \begin{align}
        {\bfP}_{\mmm}
            =   \diag{\boldsymbol{s}_{\bfR}} \bfgamma \boldsymbol{s}_{\bfT}^{\top}.
    \end{align}
    The MMM-enabled channel matrix corresponding to the above variance profile is then recovered as shown below:
    \begin{align}
    \widetilde{\bfH}
        =   \bfU_{\bfR} \left[ 
            {\bfP}_{\mmm}^{\circ \frac{1}{2}} \odot \widetilde{\bfH}_{\iid}
        \right] \bfU_{\bfT}^{\tph}.
    \end{align}
    %s
\end{itemize}

While the exact moments (or cumulant) related to $\bfW$ are intractable and, to the best of our knowledge, have not yet been reported in the technical literature, they are essential for the proposed matching procedure.
    In the next section, we address the problem of obtaining these exact moments and cumulants.

\subsection{Moments and Cumulants of the Normalized Channel Gain Matrix}
\label{sec:moment_of_matrix}

The objective of this section is to derive the eigenvalue cumulants, denoted as $\skmul_n[\uplambda_{\bfW}]$, of the normalized channel gain matrix ($\bfW$), which constitute the key quantities required by the proposed MMM framework. To this end, the derivation proceeds according to the following steps:
\begin{itemize}
    \item[-] Derive the matrix-valued moments ($\mean[\bfW_r^n]$) of the individual Wishart components ($\{ \bfW_r\}$) based on Lemma \ref{lem:exact_moment_Wk}.
    \item[-] Derive the corresponding matrix-valued cumulants ($\kmul_n[\bfW_r]$) based on Lemma \ref{lem:moment_cumulant_Wishart}.
    \item[-] Derive the matrix-valued cumulants ($\kmul_n[\bfW]$) of $\bfW = \sum_{r=1}^{R} \bfW_r$ from $\{ \kmul_n[\bfW_r]\}$ based on Theorem \ref{theo:mat_cumulant_W}.
    \item[-] Recover the matrix-valued moments ($\mean[\bfW^n]$) from the matrix-valued cumulants ($\kmul_n[\bfW]$) based on Lemma \ref{lem:moment_cumul_Wtol}.
    \item[-] Derive the eigenvalue moments ($\mean[\uplambda_{\bfW}^n]$) and the corresponding eigenvalue cumulants ($\kmul_n[\uplambda_{\bfW}]$), which are subsequently used in the moment matching procedure developed in Section \ref{sec:approximate_complex_wishart}.
\end{itemize}

\begin{Lemma}
\label{lem:exact_moment_Wk}
The $n^{\uth}$ order matrix-valued moment of the complex Wishart matrix $\bfW_r$ with one degree of freedom and a diagonal covariance ${\bfSigma}_{r}$, is formulated as follows:
\begin{align}
\mean[\bfW_r^n]
    =  \sum_{j=1}^{n}{ \frac{(n-1)!}{(n-j)!} c^{(n-j)}_r {\bfSigma}_{r}^j },
\label{eq:exact_moment_Wk}
\end{align}
where the scalar coefficient $c^{(n-j)}_k $ is derived as:
\begin{align}
c^{(n-j)}_r 
    &\triangleq \smean[\trace{\bfW_r}^{n-j}] \label{eq:ck_trace} \\
    &=  \sum_{i=1}^{n-j} \frac{(n-j-1)!}{(n-j-i)!} \trace*{{\bfSigma}_{r}^{i}} c^{(n-j-i)}_r. \label{eq:ck_recursive}
\end{align}
\end{Lemma}

\begin{IEEEproof}
The detailed proof is provided in Appendix \ref{apx:lem:exact_moment_Wk} and is based on the combinatorial framework established in \cite{graczyk2003complex}.
\end{IEEEproof}

Although Lemma \ref{lem:exact_moment_Wk} provides the exact expression of $\mean[\bfW_r^n]$, our objective is to characterize the statistics of $\bfW = \sum_{r=1}^{R}{ \bfW_{r} }$.
    Since the moments of $\bfW$ cannot be readily obtained from $\bfW_{r}$, we next derive the corresponding matrix-valued cumulants of $\bfW_r$, denoted by $\kmul_{n}[\bfW_r]$.
    
\begin{Lemma}
\label{lem:moment_cumulant_Wishart}
The recursive formula that relates the $n^{\uth}$ order matrix-valued moments and cumulants of $\bfW_r$ is obtained as follows:
\begin{align}
\kmul_{n}[\bfW_r]
    =   \mean[\bfW_r^{n}] - \sum_{\ell=1}^{n-1}{ \binom{n-1}{\ell-1} \kmul_{\ell}[\bfW_r] \mean[\bfW_r^{n-\ell}] }.
\label{eq:moment_cumulant_Wishart}
\end{align}
\end{Lemma}

\begin{IEEEproof}
We define the characteristic function (CF) and the cumulant generating function (CGF) of $\bfW_r$ as follows:
\begin{align}
{\mathbf{\Phi}}_{\bfW_r}(t)
    &\triangleq  \mean\left\{ e^{\mathbbm{i} t \bfW_r} \right\}
     \stackrel{(a)}{=}  
     \sum_{n=0}^{\infty}{ \frac{(\mathbbm{i} t)^n}{n!} \mean[\bfW_r^{n}] }, \label{eq:Phi_def} \\
{\mathbf{\Psi}}_{\bfW_r}(t)
    &\triangleq  \log {\mathbf{\Phi}}_{\bfW_r}(t).
    \label{eq:Psi_def}
\end{align}

Here, it is important to distinguish between the element-wise exponential and 
    matrix exponential, defined as $e^{\bfX} \triangleq \sum_{n=0}^{\infty} \frac{1}{n!} \bfX^n$ %\cite{kumar1965expanding}
    , that we used in equality $(a)$, as well as between the element-wise logarithm and matrix logarithm, 
        defined as $\log\bfX \triangleq \sum_{k=1}^{\infty} (-1)^{k+1} \frac{(\bfX - \bfI)^k}{k}$.

From Lemma \ref{lem:exact_moment_Wk}, we deduce that $\mean[\bfW_r^{n}]$ is strictly diagonal for any order $n$. 
Therefore, ${\mathbf{\Phi}}_{\bfW_r}(t)$ and ${\mathbf{\Psi}}_{\bfW_r}(t)$ are also diagonal. 
    Moreover, the matrix-valued cumulants, defined as $\kmul_{n}[\bfW_r]
    =   \left. \frac{1}{\mathbbm{i}^n} \frac{\partial^n {\mathbf{\Psi}}_{\bfW_r}(t)}{\partial t^n} \right|_{t\to 0}$, are thus also strictly diagonal. 

Since ${\mathbf{\Phi}}_{\bfW_r}(t) = e^{\mathbf{\Psi}_{\bfW_r}(t)}$ and both matrices are diagonal, the matrix exponential coincides with the scalar exponential exactly along the diagonal entries, i.e., $[\mathbf{\Phi}_{\bfW_r}(t)]_{i, i} = e^{[\mathbf{\Psi}_{\bfW_r}(t)]_{i, i}}$. 
    Therefore, we can directly apply the classical moment-cumulant relation to each diagonal element. 
    This can be done through the use of the generalized Leibniz's rule
    \cite[Eq. (4)]{smith1995recursive}:
\begin{align}
\frac{ \partial^n [\mathbf{\Phi}_{\bfW_r}(t)]_{i, i} }{\partial t^n}
    =  \sum_{\ell=1}^{n} \binom{n-1}{\ell-1} 
    \frac{\partial^{\ell} [\mathbf{\Psi}_{\bfW_r}(t)]_{i, i} }{\partial t^{\ell}} \frac{\partial^{n-\ell} [\mathbf{\Phi}_{\bfW_r}(t)]_{i, i}}{\partial t^{n-\ell}}.
\end{align}

Evaluating the above expression as $t \to 0$ yields the moment-cumulant relation for each diagonal entry:
\begin{align}
(\mean[\bfW_r^{n}])_{i, i}
    =  \sum_{\ell=1}^{n}{ \binom{n-1}{\ell-1} (\kmul_{\ell}[\bfW_r])_{i, i} (\mean[\bfW_r^{n-\ell}])_{i, i} }.
\end{align}

Since $\mean[\bfW_r^{n}]$ and $\kmul_{n}[\bfW_r]$ are diagonal matrices with zero off-diagonal entries, dropping the indices $(i, i)$ yields the matrix form in Eq. \eqref{eq:moment_cumulant_Wishart}. 
    This completes the proof of Lemma \ref{lem:moment_cumulant_Wishart}.
\end{IEEEproof}

Importantly, from Eq. \eqref{eq:moment_cumulant_Wishart}, diagonal cumulants yield diagonal matrix-valued moments; conversely, diagonal moments guarantee diagonal cumulants. The matrix-valued moments of $\bfW$ admit an analogous recursive relation, as established below.

\begin{Lemma}
\label{lem:moment_cumul_Wtol}
The matrix-valued moments of $\bfW$ are strictly diagonal and are obtained recursively from its matrix-valued cumulants as follows:
\begin{align}
\mean\left[ \bfW^n \right]
    =  \kmul_n\left[ \bfW \right] + \sum_{\ell=1}^{n-1}{ \binom{n-1}{\ell-1} \kmul_{\ell}\left[ \bfW \right] \mean\left[ \bfW^{n-\ell} \right] }.
\label{eq:moment_cumul_Wtol}
\end{align}
\end{Lemma}
\begin{IEEEproof}
The proof is similar to the proof of Lemma \ref{lem:moment_cumulant_Wishart} since the diagonal matrices $\mean[\bfW_r^n]$ and $\kmul_n[\bfW_r]$ yield diagonal matrices $\mean[\bfW^n]$ and $\kmul_n[\bfW]$.
\end{IEEEproof}

Having derived the matrix-valued cumulants $\bfK_{n}[\bfW_r]$, we next consider the cumulants of $\bfW = \sum_{r=1}^{R}{\bfW_r}$.
    The cumulants of $\bfW$ cannot be obtained by simply summing $\{ \bfK_{n}[\bfW_r] \}$ because $\{ \bfW_r \}$ do not commute. 
By leveraging the multivariable Zassenhaus formula, the matrix-valued cumulants $\bfK_n[\bfW]$ can be expressed as the sum of an additive contribution and a residual term accounting for matrix non-commutativity, as stated in Theorem \ref{theo:mat_cumulant_W}.

\begin{Theorem}
\label{theo:mat_cumulant_W}
The $n^{\uth}$ order matrix-valued cumulant of $\bfW = \sum_{r=1}^{R}{\bfW_r}$ 
consists of the sum of the matrix-valued cumulants of $\bfW_r$ and a residual matrix, obtained as follows:
\begin{align}
\kmul_{n}[\bfW]
    =   \underbrace{\sum_{r=1}^{R}{ \kmul_{n}[\bfW_r] }}\limits_{\triangleq \kmul_{n}^{+}[\bfW]}
    + \bfTheta_{n},
\label{eq:mat_cumulant_W}
\end{align}
with $\bfTheta_{1} = \bfTheta_{2} = \mathbf{0}_{T\times T}$ and $\bfTheta_{3}$ being obtained as shown below:
\begin{align}
\bfTheta_3
    &=   \textstyle
    \sum_{1\le i < j \le R} \trace{\bfSigma_i \bfSigma_j} \bfSigma_i - \trace{\bfSigma_i} \bfSigma_i \bfSigma_j
    \notag\\
    &\quad 
    \textstyle
    +   \sum_{1\le i < j \le R} \trace{\bfSigma_i \bfSigma_j} \bfSigma_j - \trace{\bfSigma_j} \bfSigma_i \bfSigma_j.
\label{eq:bfTheta_3}
\end{align}
\end{Theorem}

\begin{IEEEproof}
The key step here is to adopt the multivariable Zassenhaus formula \cite[Eq. (2.1)]{wang2019multivariable} to expand $e^{\mathbbm{i} t \bfW}$. Details of the proof are presented in Appendix \ref{apx:mat_cumulant_W}.
\end{IEEEproof}

Following from \eqref{eq:moment_cumul_Wtol}, the eigenvalue moments and cumulants of $\bfW$ can be determined as follows:
\begin{align}
\smean\left[ \uplambda_{\bfW}^n \right]
    &=   \frac{ 1 }{ T} \trace{\mean\left[ \bfW^n \right]},
\label{eq:scalar_moment_Wtol} \\
{\skmul}_{n}[\uplambda_{\bfW}] 
    &= \smean\left[ \uplambda_\bfW^n \right] - \sum_{i=1}^{n-1}{ \binom{n-1}{i-1} \skmul_i\left[ \uplambda_\bfW \right] \smean\left[ \uplambda_\bfW^{n-i} \right] } 
    \label{eq:scalar_cumulant_Wtol}.
\end{align}

While $\bfTheta_1 = \bfTheta_2 = \mathbf{0}_{T\times T}$, the complexity of the residual terms $\bfTheta_{n}$ in \eqref{eq:mat_cumulant_W} grows factorially with $n$ for $n \ge 3$ due to the higher-order Lie polynomials involved in the Zassenhaus expansion. 
    Therefore, to maintain analytical tractability, the subsequent derivations focus on the contributions of the additive term $\kmul^{+}_{n}[\bfW]$ in \eqref{eq:mat_cumulant_W}, which forms the basis of the proposed moment matching framework.
    
Based on the framework, the first three moments and eigenvalue cumulants of $\bfW$ are listed as shown below:
\setlength{\jot}{2pt} 
\allowdisplaybreaks
\begin{align}
\mean[\bfW] 
    &=  \textstyle \sum\limits_{r=1}^{R} { \bfSigma_r }, \\
\mean[\bfW^2] 
    &=  \textstyle \left[ \sum\limits_{r=1}^{R} { \trace{\bfSigma_r} \bfSigma_r } \right] + ( \mean[\bfW] )^2, \\
\mean[\bfW^3] 
    &\propto
    \textstyle \left[
        \sum\limits_{r=1}^{R}{ \left[ \trace{\bfSigma_r}^2 + \trace{\bfSigma_r^2} \right] \bfSigma_r - \trace{\bfSigma_r} \bfSigma_r^2 + \bfSigma_r^3 } 
    \right]
    \notag\\
    &\quad\textstyle 
     +  3 \mean[\bfW] \left[ \sum\limits_{r=1}^{R}{ \trace{\bfSigma_r} \bfSigma_r } \right]
     +  (\mean[\bfW])^3 , \\
    % =======================================================
\skmul_1[\uplambda_{\bfW}]
    &=  \frac{1}{T} \trace{\mean[\bfW]}, \label{eq:1st_cumul_W} \\
\skmul_2[\uplambda_{\bfW}]
    &=   \textstyle
    \frac{1}{T} \! \left[\trace{\mean[\bfW]^2} \!+ \sum\limits_{r=1}^R \trace{\bfSigma_r}^2 \right] - (\skmul_1[\uplambda_{\bfW}])^2, \! \label{eq:2nd_cumul_W} \\
\skmul_3[\uplambda_{\bfW}]
    &\propto  \textstyle
    \frac{1}{T} \left[ 
        \trace{\mean[\bfW]^3} 
            + 3 \text{Tr}\left( \mean[\bfW] \sum\limits_{r=1}^R \trace{\bfSigma_r} \bfSigma_r \right)
    \right.
    \notag\\
    &\qquad\textstyle
    \left.
            + \sum\limits_{r=1}^R \trace{\bfSigma_r}^3 + \trace{\bfSigma_r^3}
    \right] + \frac{2}{T^3} \trace{\mean[\bfW]}^3
    \notag\\
    &\quad\textstyle
    -\! \frac{3}{T^2} \! \left[ \trace{\mean[\bfW]^2} +\! \sum\limits_{r=1}^R \trace{\bfSigma_r}^2 \right] \! \trace{\mean[\bfW]}, \! \label{eq:3rd_cumul_W}
\end{align}
where the fourth order moments and cumulants are given by Eq. \eqref{eq:matrix_4th_moment} and Eq. \eqref{eq:scalar_4th_cumulant}, respectively, at top of the next page.
Here, the sign $\propto$ means `proportional to' and is used to indicate that the expressions account solely for the contributions from the additive term $\kmul_{n}^{+}[\bfW]$. 
\begin{figure*}
\begin{align}
\mean[\bfW^4] 
    &\propto  \textstyle
    \mean[\bfW]^4 + \left[
        \sum_{r=1}^{R} \left[ \trace{\bfSigma_r}^3 + 3 \trace{\bfSigma_r} \trace{\bfSigma_r^2} + 2 \trace{\bfSigma_r^3} \right] \bfSigma_r
        \!-   \left[ 4 \trace{\bfSigma_r}^2 + \trace{\bfSigma_r^2} \right] \bfSigma_r^2
            \!+  4 \trace{\bfSigma_r} \bfSigma_r^3 + \bfSigma_r^4
    \right]
    \notag\\
    &~\textstyle
        + 4 \mean[\bfW] \! \left[ \sum_{r=1}^{R} \! \left[ \trace{\bfSigma_r}^2 \!+ \trace{\bfSigma_r^2} \right] \! \bfSigma_r \!- \trace{\bfSigma_r} \bfSigma_r^2 \!+ \bfSigma_r^3 \right] \!
        \! + 3 \! \left[ \sum_{r=1}^{R} \! \trace{\bfSigma_r} \bfSigma_r \right]^2 \!
        \!\! + 6 (\mean[\bfW])^2 \! \left[ \sum_{r=1}^{R} \! \trace{\bfSigma_r} \bfSigma_r \right] \!, \!\!\!
    \label{eq:matrix_4th_moment} \\
% ======================================
\skmul_4[\uplambda_{\bfW}]
    &\propto \textstyle
    \frac{1}{T} 
        \text{Tr}\left\{
            4 \mean[\bfW] \! \left[ \sum_{r=1}^R \! \left[ \trace{\bfSigma_r}^2 \!+ \trace{\bfSigma_r^2} \right]\!  \bfSigma_r \!- \trace{\bfSigma_r} \bfSigma_r^2 \!+ \bfSigma_r^3 \right]\! 
            \!+ 6 \mean[\bfW]^2 \! \left[ \sum_{r=1}^R \! \trace{\bfSigma_r} \bfSigma_r \right] \! 
            \!+ 3 \! \left[ \sum_{r=1}^R \! \trace{\bfSigma_r} \bfSigma_r \right]^2 \! 
        \right\}
    \notag \\
    &~\textstyle~ 
    +   \frac{1}{T} 
    \left[
        \trace{\mean[\bfW]^4} 
        + \sum_{r=1}^R \trace{\bfSigma_r}^4 - \trace{\bfSigma_r}^2 \trace{\bfSigma_r^2} - \trace{\bfSigma_r^2}^2 + 6 \trace{\bfSigma_r} \trace{\bfSigma_r^3} + \trace{\bfSigma_r^4} 
    \right]
    \notag \\
    &~\textstyle 
        - \frac{4}{T^2} 
        \left[
            \trace{\mean[\bfW]^3} + 3 \text{Tr}\left[ \mean[\bfW] \sum_{r=1}^R \trace{\bfSigma_r} \bfSigma_r \right]
            + \sum_{r=1}^R \trace{\bfSigma_r}^3 + \trace{\bfSigma_r^3}
        \right] \trace{\mean[\bfW]}
    \notag \\
    &~\textstyle 
        - \frac{3}{T^2} \left[ \sum_{r=1}^R \trace{\bfSigma_r} + \trace{\mean[\bfW]^2} \right]^2
        + \frac{12}{T^3} \left[ \sum_{r=1}^R \trace{\bfSigma_r}^2 + \trace{\mean[\bfW]^2} \right] \trace{\mean[\bfW]}^2 
        - \frac{6}{T^4} \trace{\mean[\bfW]}^4.
    \label{eq:scalar_4th_cumulant}
\end{align}
\hrulefill
\end{figure*}

\subsection{Approximation to Semi-Correlated Wishart}
\label{sec:approximate_complex_wishart}

Building on the cumulant expressions that will be derived in Corollary \ref{lem:cumul_sum_of_iso_wishart}, we define target distribution $\widetilde{\bfW}$ as a sum of statistically independent, isotropic complex Wishart matrices:
\begin{align}
\widetilde{\bfW}
    \triangleq  \sum_{n=1}^{N} \widetilde{\bfW}_{n}, 
\label{eq:semi_corr_wishart_model}
\end{align}
where $\widetilde{\bfW}_{n} \sim {\cal CW}_{q}\left( \varP_{n}, {\scaleSigma}_{n} \bfI_{q} \right)$ follows the complex Wishart distribution of dimension $q\times q$, with $\varP_n$ degrees of freedom and covariance matrix ${\scaleSigma}_{n} \bfI_{q}$.

To assist in matching to complex semi-Wishart matrices, we introduce the following Lemma.

\begin{Corollary}
\label{lem:cumul_sum_of_iso_wishart}
Let $\bfX_{1}, \dots, \bfX_{N}$ be statistically independent $q\times q$ random matrices, where $\bfX_{n} \sim \mathcal{CW}_{q}\left( p_n, \sigma_n \bfI_{q} \right)$, and define their sum as $\bfX \triangleq \sum_{n=1}^{N}{\bfX_n}$ with $\rank{\bfX} = q$, the $\ell^{\uth}$ order scalar cumulants of the unordered eigenvalues of $\bfX$ are obtained as follows:
\begin{subequations}
\begin{align}
\skmul_{1}[\uplambda_{\bfX}]
    &=  \textstyle 
    \sum\limits_{n=1}^{N}{p_n \sigma_n}, \\
\skmul_{2}[\uplambda_{\bfX}]
    &=  \textstyle q 
    \sum\limits_{n=1}^{N}{p_n \sigma_n^2}, \\
\skmul_{3}[\uplambda_{\bfX}]
    &\propto  \textstyle 
    (q^2 + 1) \sum\limits_{n=1}^{N}{p_n \sigma_n^3}, \\
\skmul_{4}[\uplambda_{\bfX}]
    &\propto \textstyle 
    (q^3 - q^2 + 5q + 1) 
    \sum\limits_{n=1}^{N}{\nu_n \sigma_n^4}, \\
    &\cdots \notag\\
\skmul_{\ell}[\uplambda_{\bfX}]
    &\propto  \textstyle
    {\varRho}_{\ell}[\uplambda_{\bfX}] 
    \sum\limits_{n=1}^{N}{p_n \sigma_n^{\ell}}, \quad \forall \ell \ge 3. \label{eq:cumul_sum_of_iso_wishart_1} 
\end{align}

Here, coefficient $\varRho_{\ell}[\uplambda_{\bfX}]$ is defined as shown below:
\begin{align}
\varRho_{\ell}[\uplambda_{\bfX}]
    &= \sum_{r=1}^{\ell} \frac{1}{ r q^r } \sum_{i=1}^{r} \binom{r}{i} (-1)^{i-1} (q i)_{\ell},
\end{align}
where $(x)_j \triangleq \frac{(x+j-1)!}{(x-1)!}$ is the Pochhammer symbol.
\end{subequations}
\end{Corollary}

\begin{IEEEproof}
See Appendix \ref{apx:cumul_sum_of_iso_wishart}. 
\end{IEEEproof}

    As a special case of Theorem \ref{theo:mat_cumulant_W}, the contribution from the additive matrix-valued cumulants of $\bfX_n$ to the $\ell^{\uth}$ order eigenvalue cumulant of $\bfX$ nicely reduces to a tractable weighted sum of monomials in $\sigma_n$.
Such a property motivates our shift from raw moment matching to cumulant matching, whose accuracy will be highlighted in Section~\ref{sec:results}.

Matching the contribution of the additive matrix-valued cumulant terms
yields the following system of equations:
\begin{align}
\sum_{n=1}^{N} {\varP}_{n} \scaleSigma^{j}_{n}
    =   \frac{ \skmul_{j}^{+}[\uplambda_{\bfW}] }{ {\varRho}_{j}[\uplambda_{\widetilde{\bfW}}] }, 
    \quad \forall j \in \{1, 2, \dots, 2N-1\},
\label{eq:original_matching_problem}
\end{align} 
where $\skmul_{j}^{+}[\uplambda_{\bfW}]$ denotes contribution of the additive term $\kmul_n^+[\bfW]$ to the $j^{\uth}$ order eigenvalue cumulant of $\bfW$.
The system in \eqref{eq:original_matching_problem} is further subjected to the following constraints:
\begin{subequations}
\label{eq:cond_WishartSyst}
\begin{empheq}[left={}]{align}
    \scaleSigma_{n} &> 0, \quad \forall n \in \{ 1, 2, \dots, N \}, \label{eq:cond_WishartSyst_b} \\
    \varP_{n} &> 0, \quad \forall n \in \{ 1, 2, \dots, N \}, \label{eq:cond_WishartSyst_a} \\
    \scaleSigma_{n'} &\ne \scaleSigma_{n}, \quad \forall n' \ne n \in \{ 1, 2, \dots, N \}, \label{eq:cond_WishartSyst_c} \\
    \tsum_{n=1}^{N}{ \varP_{n} } &= R, \label{eq:cond_WishartSyst_d}
\end{empheq}
\end{subequations}
where Eq. \eqref{eq:cond_WishartSyst_b} and Eq. \eqref{eq:cond_WishartSyst_a} guarantee the positivity of the scale matrix and degrees of freedom, Eq. \eqref{eq:cond_WishartSyst_c} enforces the uniqueness of the solutions, 
    and, Eq. \eqref{eq:cond_WishartSyst_d} is imposed to maintain the dimension of the receive covariance.

The matching problem in Eq. \eqref{eq:original_matching_problem} with the constraint Eq. \eqref{eq:cond_WishartSyst} is known as the Stieltjes moment problem \cite{golub2009matrices}.
    In this manner, we first treat $\scaleSigma_{n}$ as the roots and ${\varP}_{n}$ as the weights and we rewrite Eq. \eqref{eq:original_matching_problem} in the following canonical form:
\begin{align}
\sum_{n=1}^{N} {\varP}_{n} \scaleSigma^{\ell}_{n}
    =   \varMu_{\ell}, \quad \forall \ell \in \{0, 1, \dots, 2N-1\},
\label{eq:transformed_matching_problem}
\end{align} 
where the target moments are defined as follows:
\begin{subequations}
\begin{empheq}[left={\empheqlbrace}]{alignat=2}
    \varMu_{0} &= R, &\quad& \text{if}~\ell = 0, \\ 
    \varMu_{\ell} &= \frac{ \skmul_{\ell}^{+}[\uplambda_{\bfW}] }{ {\varRho}_{\ell}[\uplambda_{\widetilde{\bfW}}] }, &\quad& \text{if}~ \ell = 1, 2, \dots, 2 N-1.
\end{empheq}
\end{subequations}

Solvability of the system relates to the following $N\times N$ Hankel matrices:
\begin{align}
\left[ \bfDelta_N \right]_{i,j}
    &\triangleq  \varMu_{i+j-2} 
    =   \sum_{n=1}^{N} {\scaleSigma_n^{i-1} \varP_n \scaleSigma_n^{j-1} }, \label{eq:48} \\
\left[ \mathbf{\Omega}_N \right]_{i,j}
    &\triangleq  \varMu_{i+j-1} 
    =   \sum_{k=1}^{N} {\scaleSigma_k^{i-1} \varP_k \scaleSigma_k \scaleSigma_k^{j-1} }, \label{eq:49}
\end{align}
where 
    ${\bfDelta_0 = 1}$.

Let us denote 
    ${\boldsymbol{\sigma}} \triangleq \left[ \scaleSigma_1, \scaleSigma_2, \dots, \scaleSigma_N \right]^{\top}$,
    ${\boldsymbol{p}} \triangleq \left[ \varP_1, \varP_2, \dots, \varP_N \right]^{\top}$ and  
    ${\boldsymbol{\mu}} = \left[ \varMu_{0}, \varMu_{1}, \dots, \varMu_{2N-1} \right]^{\top}$. The method for solving Eq. \eqref{eq:transformed_matching_problem} is based on the construction of the following orthogonal polynomial $\pi_N(\scaleSigma)$ \cite{golub2009matrices} as:
\begin{align}
\pi_{N}(\scaleSigma)
     =  \frac{ \det\begin{bmatrix}
        \multicolumn{4}{c}{\scalebox{1.0}{$\bfDelta_{N}$}} \\
        1 & \scaleSigma & \cdots & \scaleSigma^N
    \end{bmatrix} }{ \det\left[ \bfDelta_{N} \right] } 
    &=  \scaleSigma^N + \sum_{i=0}^{N-1}{ c_{i} \scaleSigma^{i} } \\
    &=  \prod_{n=1}^{N}\left( \scaleSigma - \scaleSigma_{n} \right).
    \label{eq:PN_ortho}
\end{align}
for $i = 1, 2, \dots, N$ and $j = 1, 2, \dots, N$. 
    The roots of $\pi_{N}(\scaleSigma)$ are eigenvalues of the symmetric tridiagonal Jacobi matrix \cite[Theorem 2.13]{golub2009matrices}:
\begin{align}
J_N = 
\begin{bmatrix}
    \alpha_{1} & \sqrt{\gamma_{1}} & 0 & \cdots & 0 \\
    \sqrt{\gamma_{1}} & \alpha_{2} & \sqrt{\gamma_{2}} & \cdots & 0 \\
    0 & \sqrt{\gamma_{2}} & \alpha_{3} & \ddots & \vdots \\
    \vdots & \vdots & \ddots & \ddots & \sqrt{\gamma_{N-1}} \\
    0 & 0 & \cdots & \sqrt{\gamma_{N-1}} & \alpha_{N}
\end{bmatrix},
\end{align}
where, due to \cite[Section 5.2]{golub2009matrices},
\begin{subequations}
\begin{align}
\alpha_{n} 
    &\triangleq \frac{ \det[\bfDelta'_{n}] }{ \det[\bfDelta_{n}] } 
    - \frac{ \det[\bfDelta'_{n-1}] }{ \det[\bfDelta_{n-1}] }, \\
\gamma_{n-1}  
    &\triangleq 
    \frac{ \det[\bfDelta_{n}] }{ \det[\bfDelta_{n-1}] } 
    \frac{ \det[\bfDelta_{n-2}] }{ \det[\bfDelta_{n-1}] },
\nonumber
\end{align}
\end{subequations}
for $n = 1, 2, \dots, N$, with $\gamma_0 = 0$.
Matrix $\bfDelta'_n \in \mathbbm{R}^{n\times n}$ is defined as follows:
\begin{align}
\left[ \bfDelta'_n \right]_{i,j}
    &\triangleq \begin{cases}
        \varMu_{i+j-2}, & \text{if } j \le n-1, \\
        \varMu_{i+j-1}, & \text{if } j = n,
    \end{cases}
\end{align}
with $\bfDelta'_{0} = 0$ and $\bfDelta'_{1} = \varMu_1$.

    Several important properties are outlined below.
\begin{itemize}
    \item If $\bfDelta_N$ is positive definite, the Sylvester’s criterion guarantees that $\bfDelta_0, \bfDelta_1, \dots, \bfDelta_{N-1}$ are also positive definite \cite[Theorem 7.2.5]{horn2012matrix}. 
    Consequently, subdiagonal elements $\sqrt{\gamma_1}, \sqrt{\gamma_2}, \dots, \sqrt{\gamma_{N-1}}$ are positive, which ensures that eigenvalues of $J_N$ are pairwise distinct.
    This means that entries of $\boldsymbol{\sigma}$ are pairwise distinct, implying that Eq. \eqref{eq:cond_WishartSyst_c} is satisfied.
    \item The $N\times N$ Vandermonde matrix $\bfV$, defined as 
    $\left[ \bfV \right]_{i, j} = \scaleSigma_j^{i-1}$,
    is invertible if and only if entries of $\boldsymbol{\sigma}$ are pairwise distinct. 
    Said differently, a sufficient condition for $\bfV$ to be  invertible is that $\bfDelta_N$ is positive definite.
    \item The Hankel matrices in Eq. \eqref{eq:48} and Eq. \eqref{eq:49} can be rewritten as follows:
    \begin{align}
        \bfDelta_N &= \bfV \diag{\ibf{p}} \bfV^{\top}, \\
        \mathbf{\Omega}_{N} &= \bfV \diag{\ibf{p}} \diag{\ibf{\sigma}} \bfV^{\top},
    \end{align}
    respectively. According to the Sylvester's law of inertia \cite{horn2012matrix}, 
    if $\bfV$ is invertible, the number of positive, negative, and zero eigenvalues of $\bfDelta_N$ and $\mathbf{\Omega}_{N}$ equals that of $\diag{\boldsymbol{p}}$ and $\diag{\boldsymbol{p}} \diag{\boldsymbol{\sigma}}$, respectively. 
\end{itemize}

\textbf{\textit{Solvability condition}}: 
    Based on the aforementioned properties, {a sufficient condition for the system to be solvable is that $\bfDelta_N$ and $\mathbf{\Omega}_{N}$ are positive definite.} 
    
    Once this condition is satisfied, built-in routines available in popular computing software (e.g., MATLAB's built-in function \texttt{eig} or Mathematica's \texttt{Eigenvalues} function) can be used to obtain eigenvalues of $J_N$, thus also yielding $\{ \scaleSigma_n \}$.
Accordingly, entries of $\boldsymbol{p}$ are derived as follows:
\begin{align}
    \boldsymbol{p} = \bfV^{-1} 
    \begin{bmatrix}
        \varMu_{0} & \varMu_{1} & \cdots & \varMu_{N-1}
    \end{bmatrix}^{\top}.
\end{align}

\textbf{\textit{Rounding and scale correction}}:
As established by the moment matching system, the calculated weights $\boldsymbol{p}$ and roots $\boldsymbol{\sigma}$ represent the degrees of freedom (d.o.f.) and scale parameters of the target Wishart matrices $\{ \widetilde{\bfW}_{n} \}$, respectively. 
    Since standard Wishart distributions are frequently defined with integer-valued d.o.f., 
    following the approach previously used in \cite{Pivaro2017TVT} for isotropic Wishart matrices, 
    we round the calculated d.o.f. to the nearest integer.
Moreover, as $\smean[\uplambda_{\widetilde{\bfW}_n}] = \scaleSigma_{n} \varP_{n}$, updating the d.o.f. requires a corresponding adjustment to the scale parameters.

With the adjusted parameters, we can now construct the target covariance matrix $\bfSigma_{\bfgamma}$ introduced in Section \ref{sec:rank_1_Approx_via_mmm} as follows:
\begin{align}
\bfSigma_{\bfgamma}
    =  \textnormal{blkdiag}\bigg(\!
        \left[ \frac{ \varP_{i} \scaleSigma_{i} }{ \nint{\varP_{i}} }
            \bfI_{ \nint{\varP_{i}} } \right]_{i=1, \dots, N} 
    \bigg), 
\end{align}
where $\nint{x}$ denotes the nearest integer to $x$.

\section{Ergodic Capacity of the Approximated Separable Channels}
Considering the rank-$1$ factorization of the power coupling matrix ${\bfP} \approx \mathbf{u} \mathbf{v}^{\top}$, the approximate  Ergodic Capacity (EC) of the Weichselberger channel is derived as follows:
\begin{align}
\widetilde{C}_{\textnormal{erg}}\left( {\mathbf{u}}, {\mathbf{v}} \right)
    =   \smean\left[ 
        \log_2\det\left( 
            \bfI_{T} + \frac{\avgsnr}{T} \bfD_{\mathbf{v}} {\bfH}_{\iid}^{\tph} \bfD_{\mathbf{u}} {\bfH}_{\iid}
        \right)
    \right],
\label{eq:ergodic_cap}
\end{align}
where 
    $\avgsnr$ denotes the total average receive SNR, which accounts for total transmission power, path loss, and noise power. 
The factors $\mathbf{u}$ and $\mathbf{v}$ for each presented rank-$1$ factorization are summarized in Table~\ref{tab:factors}.

\begin{table}[t]
\centering
\caption{Rank-$1$ factors $\mathbf{u},\mathbf{v}$ for each rank-$1$ factorization $\bfP \approx \mathbf{u}\mathbf{v}^{\top}$.}
\label{tab:factors}
\begin{tabular}{l c c}
\toprule
Factorization & $\mathbf{u}$ & $\mathbf{v}$ \\
\midrule
Kronecker & $\frac{1}{T}\, \boldsymbol{\sigma}_{\bfR}$ & $\frac{1}{R}\,  \boldsymbol{\sigma}_{\bfT}$ \\
KLD       & $\bfalpha^{\star}$ & $\bfbeta^{\star}$ \\
MMM       & $\diag{\boldsymbol{s}_{\bfR}} \bfgamma$ & $\boldsymbol{s}_{\bfT}$ \\
\bottomrule
\end{tabular}
\end{table}

\begin{Theorem}
\label{lem:exact_Ergodic}
Without loss of generality, let entries of $\mathbf{u}$ and $\mathbf{v}$ be ordered as 
${u}_{1} > {u}_{2} > \cdots > {u}_{R}$ and ${v}_{1} > {v}_{2} > \cdots > {v}_{T}$, respectively. 
    Then, Eq.\eqref{eq:ergodic_cap} can be obtained as follows:
\begin{align}
\widetilde{C}_{\textnormal{erg}}\left( {\mathbf{u}}, {\mathbf{v}} \right)
    =  - \frac{R-1}{\ln 2} + \trace*{ \widetilde{\bfPsi}^{-1} \partial\widetilde{\bfPsi} },
\label{eq:exact_Ergodic}
\end{align}
where $\widetilde{\bfPsi} \triangleq \begin{bmatrix} \widetilde{\bfPsi}_{1}^{\top} & {\bfPsi}_{2}^{\top} \end{bmatrix}^{\top}$, 
     $\partial\widetilde{\bfPsi} \triangleq \begin{bmatrix} \partial\widetilde{\bfPsi}_{1}^{\top} & \partial{\bfPsi}_{2}^{\top} \end{bmatrix}^{\top}$, with
\begin{align}
\left[ \widetilde{\bfPsi}_{1} \right]_{i, j}
    &\triangleq 
        \left( - \frac{\avgsnr {u}_j}{T} \right)^{R-1} {v}_i^{T-1} \xi_{R-1}
    \notag\\
    &\qquad
        +   \sum_{k=R-T}^{R-2} \left( - \frac{\avgsnr {u}_j}{T} \right)^{k} {v}_i^{k-R+T} \xi_{k}, \\
% -----------------------------------------------------------------------------------
\left[ \partial\widetilde{\bfPsi}_{1} \right]_{i, j}
&\triangleq 
    % -   \frac{1}{\ln 2} 
        - \frac{1}{\ln 2} \left( - \frac{\avgsnr {u}_j}{T} \right)^{R-1} {v}_i^{T-1} \xi_{R-1}
    \notag\\
&\qquad\qquad\times
        \left[ H_{R-1} - e^{ \frac{T}{\avgsnr {u}_j {v}_i} } E_1\left( \frac{T}{\avgsnr {u}_j {v}_i} \right) \right]
    \notag\\
&\qquad
        - \frac{1}{\ln 2} \sum_{k=R-T}^{R-2} \left( - \frac{\avgsnr {u}_j}{T} \right)^{k} {v}_i^{k-R+T} \xi_{k} H_{k},
\end{align}
for $i = 1, \dots, T$ and $j = 1, \dots, R$, where 
    $H_k \triangleq \sum_{r=1}^{k} \frac{1}{-R+r}$, 
    $\xi_{r} \triangleq \left( - R + 1 \right)_{r}$, and where
    $E_n(x)$ denotes the $n^{\uth}$ order exponential integral function and:
\begin{align}
% -----------------------------------------------------------------------------------
\left[ {\bfPsi}_{2} \right]_{i, j}
    &\triangleq 
    \left( - \frac{\avgsnr {u}_j}{T} \right)^{i-1} \xi_{i-1}, \\
% -----------------------------------------------------------------------------------
\left[ \partial{\bfPsi}_{2} \right]_{i, j}
    &\triangleq 
    - \frac{1}{\ln 2} \left( - \frac{\avgsnr {u}_j}{T} \right)^{i-1} \xi_{i-1} H_{i-1},
\end{align}
for $i = 1, \dots, R-T$ and $j = 1, \dots, R$.
\end{Theorem}
\begin{IEEEproof}
According to \cite{kiessling2004exact}, the moment generating function (MGF) of the mutual information (MI) 
    of double-correlated Rayleigh channels with the transmit and receive covariance matrices being $\bfD_{\mathbf{v}}$ and $\bfD_{\mathbf{u}}$, respectively, is given by the following equation:
\begin{align}
\mathcal{Z}(s) = 
    \frac{ \varrho }{ \psi_{R}(s) }
    \det\begin{bmatrix} \widetilde{\bfPsi}_{1}(s) \\ {\bfPsi}_{2}(s) \end{bmatrix},
\end{align}
where, for the sake of brevity, details of $\varrho$, $ \psi_{R}(s)$, $\widetilde{\bfPsi}_{1}(s)$, and ${\bfPsi}_{2}(s)$ are given by \cite{kiessling2004exact} and are omitted in this proof.
    Of note, $\widetilde{\bfPsi}_{1} = \widetilde{\bfPsi}_{1}(0)$, 
        $\partial\widetilde{\bfPsi}_{1} = \left. \frac{\partial}{\partial s} \widetilde{\bfPsi}_{1}(s) \right|_{s\to 0}$,
        ${\bfPsi}_{2} = {\bfPsi}_{2}(0)$, and
        $\partial{\bfPsi}_{2} = \left. \frac{\partial}{\partial s} {\bfPsi}_{2}(s) \right|_{s\to 0}$.
Accordingly, the approximated EC of the Weichselberger channel is derived as follows:
\begin{align}
\widetilde{C}_{\textnormal{erg}}\left( \mathbf{u}, \mathbf{v} \right)
    &=  \left. \frac{\partial}{\partial s} \mathcal{Z}(s) \right|_{s\to 0},
    \label{eq:CW_proof} \\
    &=  \frac{ \varrho }{ \psi_{R}(0) } 
    \!\left[ 
        \left. \frac{\partial}{\partial s} \det{\begin{bmatrix} \widetilde{\bfPsi}(s) \end{bmatrix}} \right|_{s\to 0}  
        \! - \det{\begin{bmatrix} \widetilde{\bfPsi}(0) \end{bmatrix}} \frac{ \psi'_{R}(0) }{ \psi_{R}(0) }
    \right]. \notag
\end{align}

Since $\frac{ \psi_{R}(0) }{ \varrho }=   \det{\begin{bmatrix} \widetilde{\bf\Psi}(0) \end{bmatrix}} $ as $\mathcal{Z}(0) = 1$,  
    and $\psi'_{R}(0) = \frac{R-1}{\ln{2}} \psi_{R}(0)$ \cite[Eq. (36)]{kiessling2004exact}, 
    we further derive Eq. \eqref{eq:CW_proof} as:
\begin{align}
\widetilde{C}_{\textnormal{erg}}\left( \mathbf{u}, \mathbf{v} \right)
    &=  \frac{ \left. \frac{\partial}{\partial s} \det{\begin{bmatrix} \widetilde{\bfPsi}(s) \end{bmatrix}} \right|_{s\to 0} }{ \det{\begin{bmatrix} \widetilde{\bfPsi} \end{bmatrix}} } 
        -   \frac{R-1}{\ln{2}}.
    \label{eq:C_ergTR_proof}
\end{align}

Unlike the result in \cite{kiessling2004exact}, in the present study, we adopt \cite[Eq. (9)]{golberg1972derivative} for differentiating a determinant as: 
\begin{align}
\frac{\partial}{\partial s} \det\left[ \widetilde{\bfPsi}(s) \right] 
    =   \det\left[ \widetilde{\bfPsi}(s) \right] \trace*{\widetilde{\bfPsi}(s)^{-1} \partial\widetilde{\bfPsi}(s)},
\end{align}
then plugging back in to Eq. \eqref{eq:CW_proof}, which yields Eq. \eqref{eq:exact_Ergodic}. 
    This completes the proof of Theorem \ref{lem:exact_Ergodic}.
\end{IEEEproof}

\textbf{\textit{Capacity computation and comparison:}}
We now assume that $R = T = K$ on the channel to aid in capacity analysis in high-SNR regime. 
    This assumption can be relaxed using advanced RMT techniques that are beyond the scope of the present paper.
    
In the high-SNR extreme, where $\avgsnr \to \infty$ and $K\to\infty$, 
    the EC gap between the Weichselberger channel and any separable channel with transmit and receive covariances $\bfD_{\mathbf{u}}$ and $\bfD_{\mathbf{v}}$, respectively, converges to the following:
\begin{align}
\Delta {C}_{\textnormal{erg}}\left( \mathbf{u}, \mathbf{v} \right)
\to 
    & \underbrace{\smean\!\left[ \log_2\det\left( {\bfH}^{\tph} {\bfH} \right) \right]\!
    - \smean\!\left[ \log_2\det\left( {\bfZ}^{\tph} {\bfZ} \right) \right]\!}\limits_{\Delta {C}_{\textnormal{erg}, \iid}}
\notag\\
    & - \log_2\det\left( \bfD_{\mathbf{u}} \bfD_{\mathbf{v}} \right),
\label{eq:erg_gap_uv}
\end{align}
where $\Delta {C}_{\textnormal{erg}, \iid}$ (bps/Hz) denotes the EC gap between the Weichselberger channel model and the spatially uncorrelated Rayleigh fading model.

\begin{Lemma}
\label{lem:erg_gap_kld}
In the high SNR extreme, where $\avgsnr\to\infty$ and $K\to\infty$,
    the EC gap between the Weichselberger model the KLD-enabled rank-$1$ spatially correlated Rayleigh model can be obtained as follows:
\begin{align}
\Delta {C}_{\textnormal{erg}, \kld}
    &\to
    - \frac{1}{K} \sum_{i=1}^{K} \sum_{j=1}^{K} \log_2\left( \frac{P_{i,j}}{\bar{R}_{i} \bar{C}_{j}} \right)
    \notag\\
    &\qquad
        - \frac{1}{2} \sum_{k=1}^{K} \log_2\left( \bar{R}_{k} \bar{C}_{k} \right)
        - \frac{ J^{\star} }{ K\ln 2 } \ge 0,
    \label{eq:erg_gap_kld}
\end{align}
where $\bar{R}_{i} \triangleq \frac{\sum_{k=1}^{K}{ {P}_{i, k} }}{K}$ and
    $\bar{C}_{j} \triangleq \frac{\sum_{k=1}^{K}{ {P}_{k, j} }}{K}$ denote the arithmetic means of the $i^{\uth}$ row and $j^{\uth}$ column of the coupling matrix $\bfP$, respectively.
\end{Lemma}
\begin{IEEEproof}
By evaluating the EC gap against the Kronecker model, we obtain the following:
\begin{align}
\Delta {C}_{\textnormal{erg}, \kron}
    \leadsto
    -\frac{1}{2} \sum_{k=1}^{K} \log_2\left( \bar{R}_{k} \bar{C}_{k} \right),
\label{eq:69}
\end{align}
which is directly deduced from \cite[Theorem 3]{Raghavan2010TIT}. 
    Here, the arrow `$\leadsto$' means `converges approximately to', as the exact asymptotic behavior is unavailable \cite{Raghavan2010TIT}.
Moreover, since the Kronecker model enforces that $u_k \leftarrow \bar{R}_{k}$ and $v_k \leftarrow \bar{C}_{k}$, which is deduced from Table \ref{tab:factors}, we obtain the following:
\begin{align}
\Delta {C}_{\textnormal{erg}, \iid}
\leadsto
    \sum_{k=1}^{K} \log_2\left( \bar{R}_{k} \bar{C}_{k} \right).
\label{eq:erg_gap_iid}
\end{align}

To obtain the desired EC gap, we first deduce from Eq. \eqref{eq:14} that 
    $\alpha_i^{(\kappa)} \!= \frac{1}{K} \sum_{j=1}^{K}{ \frac{ P_{i,j} }{\beta_j^{(\kappa)}} }$ and 
    $\beta_j^{(\kappa)} \!= \frac{1}{K} \sum_{i=1}^{K}{ \frac{ P_{i,j} }{\alpha_i^{(\kappa)}} }$ for iteration~$\kappa$. 
Combining the aforementioned equalities yields the following:
\begin{align}
    K^2 - \sum_{i=1}^{K} \sum_{j=1}^{K} \frac{ P_{i, j} }{ \alpha_i^{\star} \beta_j^{\star} } = 0.
\end{align}

Next, we rewrite objective function $J^{\star} = J\left( \bfalpha^{\star}, \bfbeta^{\star} \right)$ initially defined in Eq. \eqref{eq:12}.
By applying the property derived above, we obtain the following:
\begin{align}
\log_2\det\left( \bfD_{\bfalpha^{\star}} \bfD_{\bfbeta^{\star}} \right)
    &=  \frac{1}{K} \sum_{i=1}^{K} \sum_{j=1}^{K} \log_2 \left( \alpha_i^{\star} \beta_j^{\star} \right) \\
    &=  \frac{ J^{\star} }{K \ln 2} + \frac{1}{K} \sum_{i=1}^{K} \sum_{j=1}^{K} \log_2 P_{i, j}.
\label{eq:log2_ui_vj} 
\end{align}

Plugging Eq. \eqref{eq:erg_gap_iid} and Eq. \eqref{eq:log2_ui_vj} back into the initial EC gap in Eq. \eqref{eq:erg_gap_uv} and performing the necessary algebraic reductions, we obtain Eq. \eqref{eq:erg_gap_kld} and conclude the proof of Lemma \ref{lem:erg_gap_kld}.
\end{IEEEproof}

Applying non-negativity property $J^{\star} \ge 0$ to Eq. \eqref{eq:erg_gap_kld} yields a lower bound for the EC gap between the Kronecker channel and the KLD-enabled channel, which is obtained as follows:
\begin{align}
\Delta {C}_{\textnormal{erg}, \kron}
-   \Delta {C}_{\textnormal{erg}, \kld}
\ge \frac{1}{K} \sum_{i=1}^{K} \sum_{j=1}^{K} 
    \log_2 \left( \frac{ P_{i, j} }{ \bar{R}_{i} \bar{C}_j } \right),
\label{eq:erg_gap_kld_kron}
\end{align}
where the equality occurs when $J^{\star} = 0$, which is satisfied if and only if $\bfP$ is already rank-$1$. 
Moreover, the lower bound is nonnegative if the following condition is met:
\begin{align}
\left[ \prod_{i=1}^K \prod_{j=1}^K P_{i,j} \right]^{\frac{1}{K^2}} 
    \ge \left[ \prod_{i=1}^K \bar{R}_i \right]^{\frac{1}{K}} \left[ \prod_{j=1}^K \bar{C}_j \right]^{\frac{1}{K}},
\end{align}
that is, when the geometric mean of the entries of $P$ exceeds the product of the geometric means of its row and column averages. 
    When the above condition holds, a larger value on the right-hand side of Eq. \eqref{eq:erg_gap_kld_kron} indicates that the Kronecker model yields a larger capacity gap. Conversely, as established by Eq. \eqref{eq:erg_gap_kld}, the proposed KLD-enabled rank-$1$ model maintains a lower EC gap.

    In physical propagation environments, a power coupling matrix satisfying the above inequality characterizes a {\it sparse and non-regular} scattering structure, 
    where the total channel power is accumulated within a few entries of $\bfP$, while the remaining entries are negligible \cite{Raghavan2010TIT}.
Such environments cause severe EC underestimation in the Kronecker model, making the proposed KLD-enabled model a promising alternative for an accurate capacity estimation in high-SNR regimes.

\section{Numerical Results and Discussion}
\label{sec:results}

Unless stated otherwise, in this section, we consider a point-to-point MIMO system with $T = 8$ transmit and $R = 8$ receive antennas.
    We examine the following three propagation environments: sparse, intermediate, and rich scattering, where the corresponding power coupling matrices are adopted from \cite[Eq. (60)-(62)]{Raghavan2010TIT}.
The corresponding matrices $\bfP$ are normalized so that the total average channel power satisfies $P_{\bfH} = RT = 64$, and the rank-$1$ KLD-enabled factorization of Algorithm \ref{alg:optimal_rank1_nmf} is performed with a convergence tolerance of $\epsilon = 10^{-12}$.
    The weighting matrices for the moment matching procedure are set to $\boldsymbol{s}_{\bfR} = \bfalpha^{\star}$ and $\boldsymbol{s}_{\bfT} = \boldsymbol{\sigma}_{\bfT}$, which are heuristically chosen based on numerical experimentation.

\subsection{Convergence of the KLD-Enabled Factorization}

\begin{figure}
    \centering
    \includegraphics[width=0.8\linewidth]{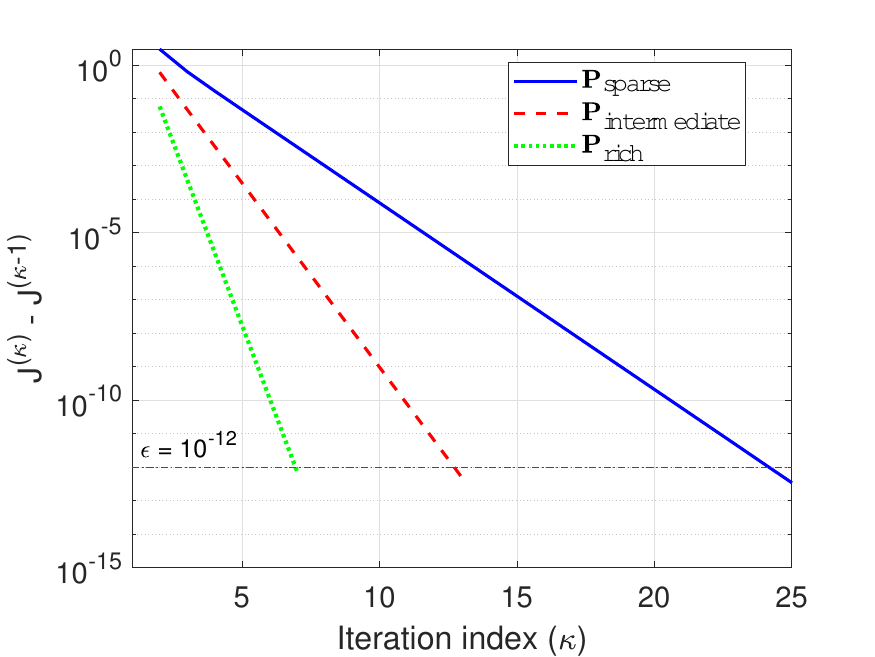}
    \caption{Convergence of $J^{(\kappa-1)} - J^{(\kappa)}$ for the KLD-enabled rank-$1$ factorization across sparse, intermediate, and rich scattering environments, where $J^{(\kappa)} = J\left( \bfalpha^{(\kappa)}, \bfbeta^{(\kappa)} \right)$.}
    \label{fig:figure1_kld}
\end{figure}

Fig. \ref{fig:figure1_kld} shows the convergence rate of the KLD-enabled rank-$1$ factorization presented by Algorithm \ref{alg:optimal_rank1_nmf} under the three considered scattering environments, defined by $\bfP_{\textsf{sparse}}$, $\bfP_{\textsf{intermediate}}$, and $\bfP_{\textsf{rich}}$. 
    As can be seen in Fig. \ref{fig:figure1_kld}, objective function $J\left( \bfalpha^{(\kappa)}, \bfbeta^{(\kappa)} \right)$ reliably monotonically descends with each iteration $\kappa$, which validates that optimization sequence $\{ \bfalpha^{(\kappa)}, \bfbeta^{(\kappa)} \}$ reliably descends towards a stationary point without fluctuation. 
In fact, the algorithm demonstrates rapid convergence in fewer than ${\bf 25}$ iterations across all scenarios, even under the strict tolerance error of $\epsilon = 10^{-12}$. 
    The convergence rate is governed by the structure of $\bfP$: the rich scattering environment, whose coupling matrix is nearly rank-1, is well-conditioned and reaches the tolerance $\epsilon = 10^{-12}$ in roughly $7$ iterations, whereas the sparse environment, characterized by highly concentrated power entries and a coupling matrix far from rank-$1$, converges more slowly, requiring about $24$ iterations. 
In all cases, the strict tolerance is met in fewer than $25$ iterations, combined with the $O(RT)$ per-iteration cost, confirms the computational efficiency of the proposed factorization.

\subsection{Ergodic Capacity Accuracy and Scalability}

\begin{figure}[!t]
    \centering
    \subfloat[]{
        \includegraphics[width=0.9\linewidth]{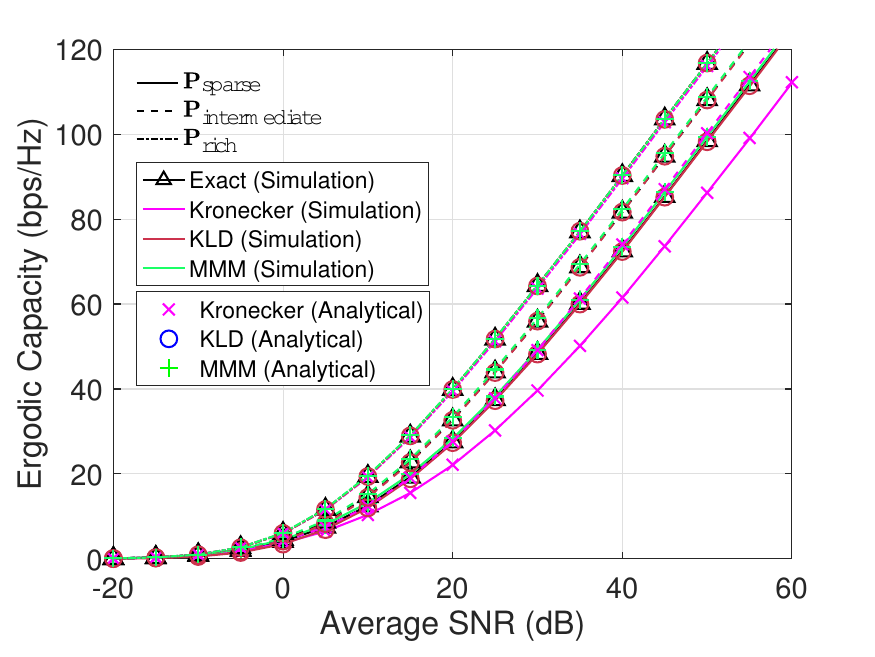}
        \label{fig:figure3_ergodic_capacity}
    }
    \\
    \subfloat[]{
        \includegraphics[width=0.9\linewidth]{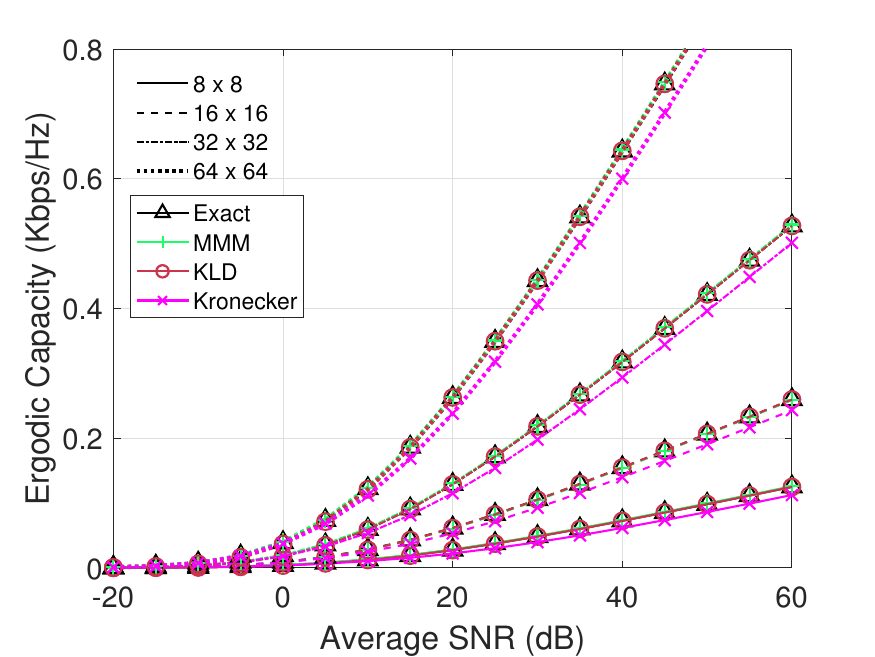}
        \label{fig:figure7_asymptotic}
    }
    \caption{Ergodic capacity vs. average SNR ($\bar{\gamma}$): (a) across sparse, intermediate, and rich scattering environments; (b) capacity scaling for MIMO antenna array dimensions in a sparse scattering environment. \label{fig:figure2_ergodic_capacity}}
\end{figure}

Fig.~\ref{fig:figure2_ergodic_capacity}(a) plots the EC against the average SNR ($\bar{\gamma}$) under the three examined environments 
to evaluate the accuracy of the proposed rank-$1$ approximations.
    As can be seen in the figure, the derived analytical results of Theorem~\ref{lem:exact_Ergodic} (illustrated by colored markers) perfectly match the simulation curves across all SNR regimes. 
    This validates the derived exact closed-form capacity expressions.
In line with \cite{Raghavan2010TIT}, as entries of $\bfP$ become more sparse, the traditional Kronecker model severely underestimates the true capacity under medium-to-high SNR regimes.
    In contrast, the proposed KLD-enabled rank-$1$ model matches the exact Weichselberger capacity curves across all SNR regimes, though slight mismatches are observed when using the MMM-enabled rank-$1$ model.
Yet, both approximations demonstrate that rank-$1$ models remain highly promising, as they accurately capture the EC with a significantly reduced complexity. 
    To further characterize the MMM solutions, the rounded solutions to the moment matching system in~Eq. \eqref{eq:transformed_matching_problem} for the sparse, 
intermediate, and rich scattering environments are $\boldsymbol{\sigma} = \left[0.0815, 0.2420\right]^{\top}$ and 
$\boldsymbol{p} = \left[5, 3\right]^{\top}$;
$\boldsymbol{\sigma} = \left[0.0858, 0.1783\right]^{\top}$ and 
$\boldsymbol{p} = \left[4, 4\right]^{\top}$; and
$\boldsymbol{\sigma} = \left[0.1135, 0.1605\right]^{\top}$ and 
$\boldsymbol{p} = \left[6, 2\right]^{\top}$, respectively.
    This reveals that richer scattering environments cause the normalized matrix $\bfW$ to become more isotropic.

To examine whether this accuracy is preserved at larger antenna regimes, Fig.~\ref{fig:figure2_ergodic_capacity}(b) evaluates scalability of the 
proposed rank-$1$ approximations by varying the MIMO antenna array dimensions from $8\times 8$ to $64\times 64$ under the sparse scattering environment.
    As can be seen in Fig.~\ref{fig:figure2_ergodic_capacity}(b), the KLD- and MMM-enabled models consistently match the exact capacity curves across all array configurations and SNR regimes, suggesting that the proposed framework scales well with increasing system dimensions.

\subsection{Eigenvalue Moment Matching: Accuracy and the Role of the Residual $\bfTheta_n$}

\begin{figure}
    \centering
    \includegraphics[width=0.9\linewidth]{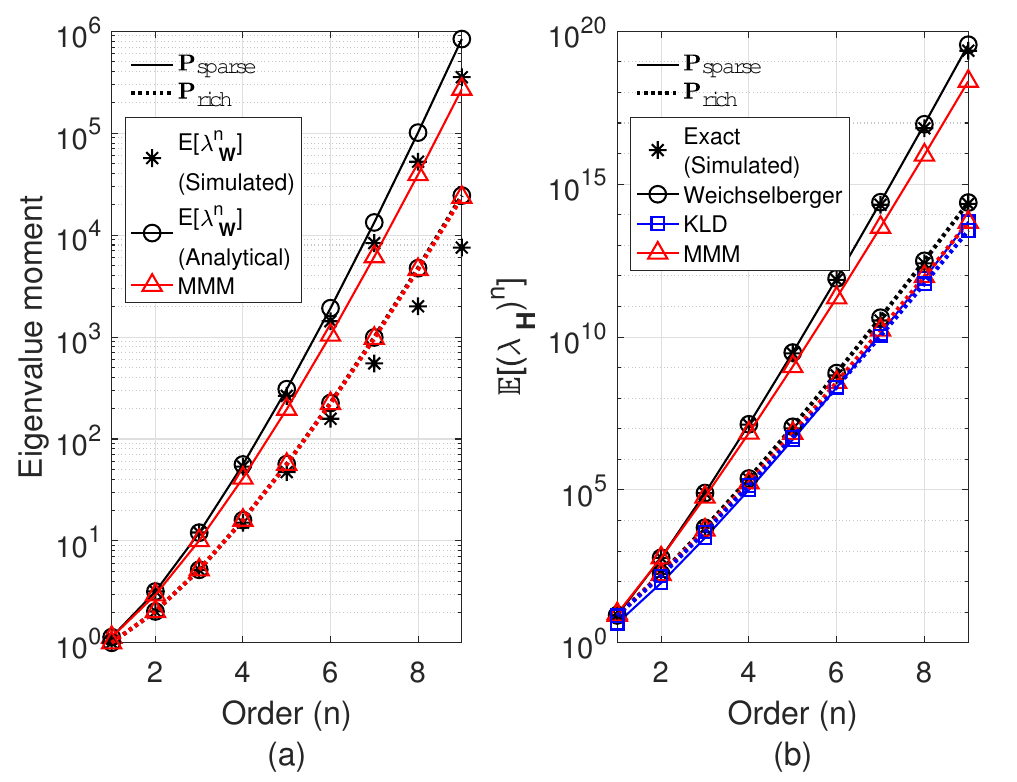}
    \caption{Illustration of the eigenvalue moments with varying orders $n$ under sparse and rich scattering conditions, comparing the  moments formulas against exact Monte Carlo simulations with $10^6$ independent realizations: 
    a) normalized channel gain $\bfW$; 
    b) Weichselberger channel gain $\bfG_{\bfH}$.}
    \label{fig:figure6_moments_verification}
\end{figure}

Fig.~\ref{fig:figure6_moments_verification} shows eigenvalue moments $\smean[\uplambda_{\bfW}^n]$ and
$\smean[\uplambda_{\bfH}^n]$ as a function of the order $n$, where the analytical 
curves (solid and dotted) are derived using Eq. \eqref{eq:exact_moment_Wk}, Eq. \eqref{eq:moment_cumulant_Wishart}, Eq. \eqref{eq:mat_cumulant_W}, Eq. \eqref{eq:moment_cumul_Wtol}, and Eq. \eqref{eq:scalar_moment_Wtol} in Section~\ref{sec:moment_of_matrix} 
while neglecting residual term $\bfTheta_n$.
    As expected in Fig. \ref{fig:figure6_moments_verification}(a), the analytical eigenvalue moments perfectly match the exact simulated moments of the target matrix $\bfW$ for the first two orders and start diverging beyond the third order moment.
Interestingly, applying integer rounding to the solutions of the moment-matching system actually brings the moments closer to the exact values in the sparse scattering environment.
    This indicates that the residual term $\bfTheta_{n}$ marginally contributes to the eigenvalue moments of $\bfG_{\bfH}$ and the overall EC.
The major finding in Fig. \ref{fig:figure6_moments_verification}(b) is that the analytical eigenvalue moments accurately match the exact eigenvalue moments of the physical channel gain $\bfG_{\bfH}$ across both environments.
    
% To further investigate the mismatch in low-to-medium SNR regimes, in what follows, we consider two unusual scattering scenarios.

\subsection{Performance Under Highly Structured Coupling}

\begin{figure}
    \centering
    \includegraphics[width=\linewidth]{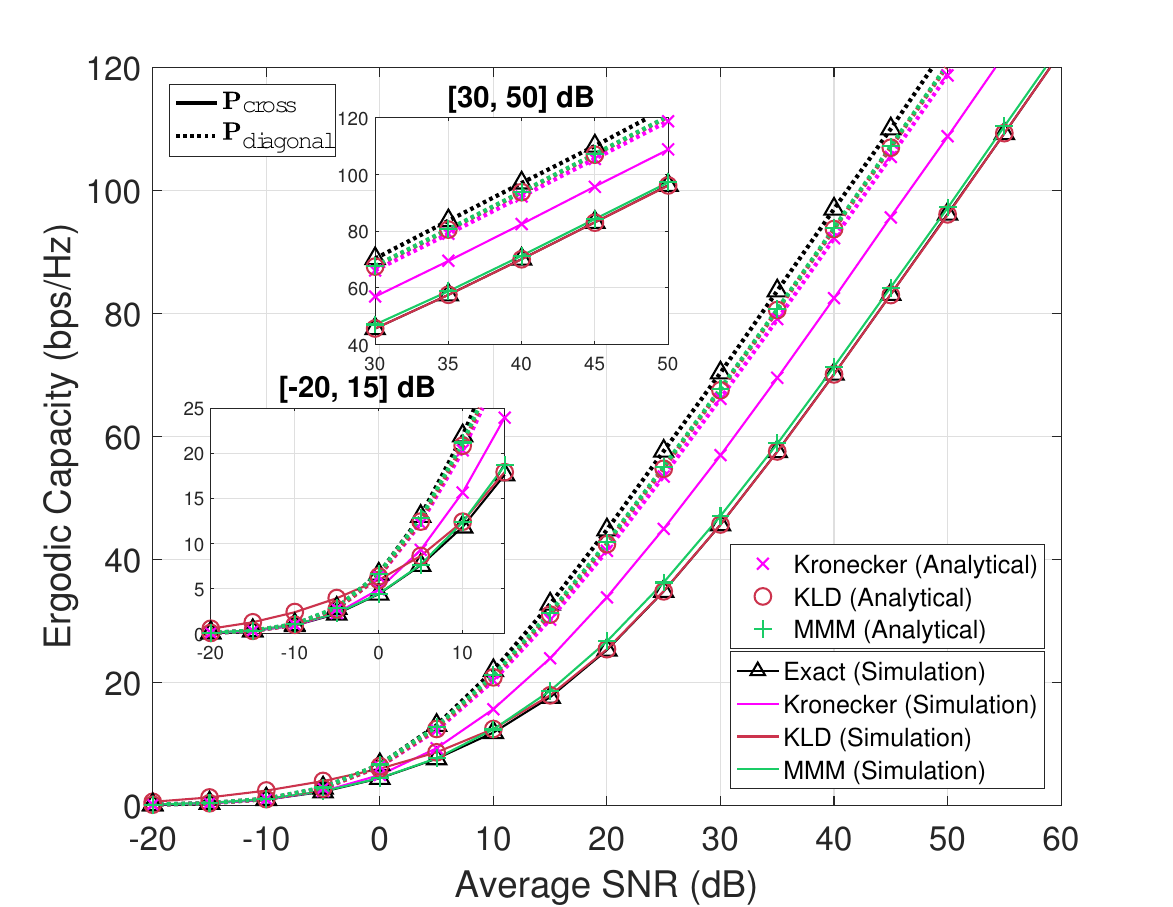}
    \caption{Ergodic capacity vs. average SNR ($\overline{\gamma}$) for highly structured coupling matrices $\bfP_{\textsf{diag}}$ and $\bfP_{\textsf{cross}}$. \label{fig:figure4_channel_regularity}}
\end{figure}
In this section, the power coupling matrices are written in the normalized form as follows:
\begin{align}
\left[ \bfP \right]_{i, j}
    =   \frac{ p_{i, j} }{ \frac{1}{RT} \sum_{r=1}^{R} \sum_{t=1}^{T} {p_{r, t}} },
\end{align}
where $\{ p_{i, j} \}$ is a set of random variables (RVs) defined on $[0, m_{i, j}]$.  
    By adjusting $m_{i, j}$ at specific values, we can generate different structured coupling matrices to model various physical radio environments.
In this section, we examine the EC for highly structured spatial power matrices previously mentioned in \cite{Weichselberger2006TWC, ren2026multi}-- namely, the cross-like and diagonal structures. 
    Once again, the traditional Kronecker model fails to match the exact EC, since it spreads out the total channel power across the entries of $\bfP$, thus inflating the true degrees of freedom (DoF) of $\bfH$. 
We now examine capacity mismatches of the KLD- and MMM-enabled models under each individual coupling structure.
A realization of each structure is shown at the top of the next~page. 

\begin{table*}[!t]
\small
\noindent
\begin{minipage}[l]{0.485\linewidth}
\begin{align}
&\scriptsize\bfP_{\textsf{cross}} = \notag\\
&\scriptsize\begin{bmatrix}
    0.0348 & 0.0191 & 6.1073 & 0.0558 & 0.0090 & 0.0680 & 0.0322 & 0.0171 \\
    0.0021 & 0.0026 & 6.8633 & 0.0495 & 0.0599 & 0.0424 & 0.0094 & 0.0232 \\
    0.2926 & 5.9513 & 7.5283 & 4.9849 & 0.5388 & 1.6490 & 4.3739 & 3.6708 \\
    0.0440 & 0.0014 & 2.9107 & 0.0463 & 0.0450 & 0.0110 & 0.0514 & 0.0328 \\
    0.0157 & 0.0489 & 6.1194 & 0.0187 & 0.0688 & 0.0297 & 0.0097 & 0.0741 \\
    0.0345 & 0.0150 & 3.3489 & 0.0697 & 0.0540 & 0.0223 & 0.0136 & 0.0286 \\
    0.0202 & 0.0063 & 2.6051 & 0.0003 & 0.0681 & 0.0546 & 0.0544 & 0.0183 \\
    0.0511 & 0.0077 & 5.3774 & 0.0732 & 0.0511 & 0.0555 & 0.0437 & 0.0134
\end{bmatrix} \label{eq:P_cross}
\end{align}
\end{minipage}%
\begin{minipage}[l]{0.51\linewidth}
\begin{align}
&\scriptsize\bfP_{\textsf{diag}} = \notag\\
&\scriptsize\begin{bmatrix}
    4.7120 & 0.0858 & 0.0325 & 0.0470 & 0.0468 & 0.0845 & 0.0940 & 0.1161 \\
    0.0261 & 8.6129 & 0.1127 & 0.0639 & 0.0027 & 0.0274 & 0.0498 & 0.0183 \\
    0.0565 & 0.0365 & 6.9320 & 0.1151 & 0.0221 & 0.0979 & 0.0961 & 0.0702 \\
    0.0938 & 0.0730 & 0.0460 & 11.7124 & 0.0549 & 0.0534 & 0.0924 & 0.0936 \\
    0.0555 & 0.0181 & 0.1121 & 0.0203 & 3.6522 & 0.0585 & 0.0860 & 0.0269 \\
    0.1110 & 0.1004 & 0.0991 & 0.1179 & 0.0254 & 1.7186 & 0.0055 & 0.0905 \\
    0.1041 & 0.0175 & 0.0902 & 0.0326 & 0.0759 & 0.1096 & 11.6444 & 0.0745 \\
    0.0715 & 0.0035 & 0.0879 & 0.0381 & 0.0072 & 0.0261 & 0.0599 & 11.4805
\end{bmatrix} \label{eq:P_diag}
\end{align}
\end{minipage}
\end{table*}

{\it 1) Cross-like coupling structure}, characterized by $\bfP_{\textsf{cross}}$ in Eq. \ref{eq:P_cross}, which produces full-rank correlation matrices at both the transmit and the receive link ends. 
    However, the instantaneous channel realizations generated from this structure are approximately rank-$2$. 
This physical behavior aims to mimic the ``keyhole channel phenomenon'', where a physical bottleneck restricts the spatial DoF despite rich local scattering at both the transmitter and the receiver. 
    Fig. \ref{fig:figure4_channel_regularity} confirms the limitation of the KLD-enabled rank-$1$ model in the low-to-medium SNR regime (e.g., $-20$ to $15$ dB interval), where capacity mismatches begin to appear.
The figure also highlights the role of MMM-enabled model, which acts as a robust alternative to the KLD-enabled model that maintains high accuracy at low SNR while causing only a slight capacity mismatch in the high SNR~region.
    
{\it 2) Diagonal coupling structure}, characterized by $\bfP_{\textsf{diag}}$ Eq. \ref{eq:P_diag}, where each transmit eigenmode is linked to a single receive eigenmode, through a single scatterer or a scattering cluster. 
    This structure aims to characterize spatial multiplexing frequently encountered in the presence of MIMO beamforming and/or RIS-enabled channels~\cite{ren2026multi}.
    A closer inspection of the $[30, 50]$ dB interval reveals a clear limitation of the rank-$1$ frameworks: specifically, the Kronecker, KLD-enabled, and MMM-enabled models all underestimate the exact EC. 
    Such a discrepancy highlights that rank-$1$ approximations cannot preserve the parallel, orthogonal eigenmodes required to support multiple independent data streams, thereby failing to capture the full spatial multiplexing gain in the high-SNR regime. 
Yet, the proposed KLD-enabled and MMM-enabled models still provide more reliable capacity estimates by yielding a narrower EC gap as compared to the traditional Kronecker model.

\section{Conclusion}

In this paper, to address the analytical complexities of the non-separable channel model, we proposed two rank-$1$ factorization frameworks for power coupling matrix $\bfP$ of the Weichselberger MIMO channels.
    To this end, we first introduced a KLD-enabled approximation computed via an efficient rank-$1$ IS-NMF under a strict error tolerance. 
Our analytical and numerical results revealed that the KLD-enabled model has a high accuracy in characterizing the EC in medium-to-high SNR regimes, especially for sparse scattering environments where the traditional Kronecker model severely underestimates performance. 
    To address the KLD model's capacity mismatches in low-SNR and highly structured scattering scenarios, we subsequently developed a novel MMM-enabled channel. 
By mapping eigenvalue moments of $\bfG_{\bfH}$ onto a Wishart distribution, we found that the MMM framework preserves the dominant statistical moments while maintaining robust capacity estimates across all SNR regimes. 
    Both proposed approximations were found to enable exact, closed-form expressions for the EC while providing highly accurate and computationally efficient alternatives to the traditional Kronecker model.

\appendices
\section{Proof of Lemma \ref{lem:exact_moment_Wk}}
\label{apx:lem:exact_moment_Wk}

The $(k_1, k_{n+1})^{\uth}$ element of the matrix-valued moment $\mean[\bfW_r^n]$ is derived by expanding $\mean[\bfW_r^n]$ using the matrix products as follows:
\begin{align}
\left( \mean[\bfW_r^n] \right)_{k_1, k_{n+1}}
    &=  \smean\Big[ \left( \bfW_r^n \right)_{k_1, k_{n+1}} \Big] \notag\\
    &=  \sum_{k_2, \dots, k_n } \smean\Big[ \tprod_{j=1}^{n} W_{r, (k_j, k_{j+1})} \Big],
\label{eq:apx_1_74}
\end{align}
where $W_{r, (i, j)}$ denotes the $(i,j)^{\uth}$ element of $\bfW_{r}$.

According to \cite[Corollary 6]{graczyk2003complex}, 
we obtain:
\begin{align}
\left( \mean[\bfW_r^n] \right)_{k_1, k_{n+1}}
    =  \sum_{\pi \in S_{n}} \sum_{k_2, \dots, k_n } \prod_{j=1}^{n}{ \Sigma_{r, (k_j, k_{\pi(j)+1})} }.
\label{eq:apx_1_76}
\end{align}
where $\Sigma_{r, (i, j)}$ denotes the $(i,j)^{\uth}$ element of $\bfSigma_r$, 
    $S_n$ refers to the set of all possible permutation of $\{ 1, 2, \dots, n\}$, and $\pi \in S_n$ is a permutation picked out from $S_n$.
As matrix $\bfSigma_r$ is diagonal, the product term in Eq. \eqref{eq:apx_1_76} is non-zero if and only if $k_j = k_{\pi(j)+1}$, $\forall j \in [1, n]$. 
    Hence, $\left( \mean[\bfW_r^n] \right)_{k_1, k_{n+1}} = 0$ for all $k_1 \ne k_{n+1}$, thus $\mean[\bfW_r^n]$ is also strictly diagonal.
    
The diagonal elements of $ \mean[\bfW_r^n]$ are rewritten as:
\begin{align}
\left( \mean[\bfW_r^n] \right)_{k_1, k_1}
    &=  \sum_{\pi \in S_{n}} \Sigma_{r, (k_1, k_1)}^{ l_1(\pi+1) } 
    \qquad\quad \prod_{ \mathclap{c \in C(\pi+1) \backslash \{ c_1(\pi+1) \}} } \quad \trace*{\bfSigma_r^{|c|}} \\
    &=   \sum_{j=1}^{n} { \Sigma_{r, (k_1, k_1)}^j } 
    \sum_{ \mathclap{ \substack{ \pi \in S_n \\ j = |c_1(\pi+1)| } } } \quad 
    \qquad\quad  \prod_{ \mathclap{c \in C(\pi+1) \backslash \{ c_1(\pi+1) \}} } \quad  
    \trace*{\bfSigma_r^{|c|}}, \label{eq:apx_79}
\end{align}
where $l_1(\pi+1) = |c_1(\pi+1)|$ denotes the length of the cycle containing the index $k_1$, i.e., $c_1(\pi+1) = (k_1, \dots, k_{n+1})$ and
    $\sum_{ \substack{ \pi \in S_n \\ j = |c_1(\pi+1)| } }$ is a shorthand notation for summation over all permutations in $S_n$ with the cycle $c_1(\pi+1)$ in $C(\pi+1)$ having the length of $j$.
Of note, the product term in Eq. \eqref{eq:apx_79} is invariant over all combinations of $c_1(\pi+1)$, each with length of $j$ and contains the starting element $k_1$.    
    We have $\binom{n-1}{j-1}$ combinations to choose the remaining $(j-1)$ elements from the available $(n-1)$ elements. 
Moreover, these $j$ elements can be arranged into a distinct cycle in $(j-1)!$ ways. 
Therefore, we rewrite Eq. \eqref{eq:apx_79} as follows:
\begin{align}
& \sum_{ \substack{\pi\in S_n \\ j = l_1(\pi+1)} } 
\prod_{ c\in C(\pi+1) \backslash \{ c_1(\pi+1) \} }
\trace*{\bfSigma_r^{|c|}}
\notag\\
&\qquad\quad
    =   \binom{n-1}{j-1}(j-1)! \sum_{\pi \in S_{n-j}} 
    \prod_{c \in C(\pi)} \trace*{\bfSigma_r^{|c|}} \label{eq:apx_1_binom_} \\
&\qquad\quad
    \mathop{=}\limits^{(a)} \frac{(n-1)!}{(n-j)!} \smean\left[ \trace*{\bfW_r}^{n-j} \right],
\end{align}
where $(a)$ is due to \cite[Theorem 1]{graczyk2003complex}.

Plugging the above result back into Eq. \eqref{eq:apx_79} and considering that $\mean[\bfW_r^n]$ is strictly diagonal, we obtain Eq. \eqref{eq:exact_moment_Wk}. 
For the recursive form in Eq. \eqref{eq:ck_recursive}, we first rewrite $c^{(t)}_{r}$ as follows:
\begin{align}
c_k^{(t)} 
    =  \sum_{i=1}^{t} 
    \sum_{ \substack{\pi \in S_m \\ |c_1(\pi)| = i} } \trace*{\bfSigma_r^{i}}
    \prod_{ c \in C(\pi) \backslash \{c_1(\pi) \} } \trace*{\bfSigma_r^{|c|}}.
\end{align}

Similarly to the approach used to obtain Eq. \eqref{eq:apx_1_binom_}, forming $c_1(\pi)$ from the available $(t-1)$ elements yields $\binom{t-1}{i-1}(i-1)! = \frac{(t-1)!}{(t-i)!}$ distinct cycles. 
    This yields:
\begin{align}
c_k^{(t)} 
    &=  \sum_{i=1}^{t} \frac{(t-1)!}{(t-i)!} \trace*{\bfSigma_r^{i}}
    \underbrace{ \sum_{ \pi \in S_{t-i} }  \prod_{ c \in C(\pi) } \trace*{\bfSigma_r^{|c|}} }\limits_{c_k^{(t-i)}}.
\end{align}
This completes the proof of Lemma \ref{lem:exact_moment_Wk}.

\section{Proof of Theorem \ref{theo:mat_cumulant_W}}
\label{apx:mat_cumulant_W}
Using the multivariable Zassenhaus formula \cite[Eq. (2.1)]{wang2019multivariable} to expand $e^{\mathbbm{i} t \bfW}$, we obtain the following series for the CGF of $\bfW$:
\begin{align}
\mathbf{\Psi}_{\bfW}(t)
    =   \log\mean\left[ e^{\mathbbm{i} t \bfW_{1}} \dots e^{\mathbbm{i} t \bfW_{R}}
    e^{(\mathbbm{i} t)^2 \bfC_{2}} e^{(\mathbbm{i} t)^3 \bfC_{3}} \dots \right],
\label{eq:Phi_W_t}
\end{align}
where $\bfC_k$ denotes the homogeneous Lie polynomial in $\{ \bfW_r \}$ of degree $k$.
    To derive the $n^{\uth}$ order moment or cumulant, one must plug in the explicit formulas of $\bfC_2, \bfC_3, \dots, \bfC_n$.
These Lie polynomials were previously presented in \cite{wang2019multivariable} with 
    $\bfC_2 \triangleq\! \frac{1}{2} \sum_{1\le i < j \le R} [\bfW_j, \bfW_i]$ and $\bfC_3$ being as follows:
\begin{align}
{\bfC}_3 
    &\triangleq
\textstyle
    \sum_{1 \le i < j \le R} 
    \frac{1}{6} [\bfW_i, [\bfW_i, \bfW_j]] 
    + \frac{1}{3} [\bfW_j, [\bfW_i, \bfW_j]]
\notag\\
    &\quad
\textstyle
    + \frac{1}{3} \sum_{1 \le i < j < k \le R} 
    [\bfW_j, [\bfW_i, \bfW_k]] 
    + [\bfW_k, [\bfW_i, \bfW_j]],
\end{align}
where $[\mathbf{X}, \mathbf{Y}] \!=\! \mathbf{X} \mathbf{Y} \!-\! \mathbf{Y} \mathbf{X}$ denotes the commutator of $\mathbf{X}$ and $\mathbf{Y}$.
    
Let us hereby specifically focus on the third order cumulant.
    First, we use the following definition of matrix exponential:
\begin{align}
e^{(\mathbbm{i} t)^2 \bfC_{2}} 
    &=  \textstyle \bfI_T + (\mathbbm{i} t)^2 \bfC_2 + O(t^{4}), \label{eq:eC2_series} \\
e^{(\mathbbm{i} t)^3 \bfC_{3}} 
    &=  \textstyle \bfI_T + (\mathbbm{i} t)^3 \bfC_3 + O(t^{4}), \label{eq:eC3_series} \\
e^{\mathbbm{i} t \bfW_{1}} \dots e^{\mathbbm{i} t \bfW_{R}}
    &= \textstyle \bfI_T 
        + \mathbbm{i} t \sum_{r=1}^{R} {\bfW_r} + O(t^2).
\end{align}

Since $\{ \bfW_r \}$ are mutually independent and $\{ \mean[\bfW_r] \}$ are also mutually commutative due to being strictly diagonal matrices,
    it is straightforward to show that $\mean[\bfC_2] = \mathbf{0}_{T\times T}$.
Plugging the above results into Eq. \eqref{eq:Phi_W_t}, we obtain:
{
\begin{equation}
\mathbf{\Psi}_{\bfW}(t)
    = \log\!\Bigg[
      \mathbf{\Pi}(t) 
      + (\mathbbm{i} t)^3 
      \Bigg[\underbrace{\mean[\bfC_3] + \tsum_{r=1}^{R} \mean[\bfW_r \bfC_2]}_{\frac{1}{3!} \bfTheta_3}
      \Bigg]
    \Bigg]\! + O(t^4) \notag
\end{equation}
}
where $\mathbf{\Pi}(t) \triangleq \prod_{r=1}^{R} \mean[e^{ \mathbbm{i} t \bfW_r }]$.

    The derivation for the closed-form expression of $\bfTheta_3$ in Eq. \eqref{eq:bfTheta_3} uses 
    $\mean[\bfW_i] = \bfSigma_i$ and
    $\mean[\bfW_i^2] = \trace{\bfSigma_i} \bfSigma_i + \bfSigma_i^2$, both obtained directly from Lemma~\ref{lem:exact_moment_Wk}, together with $\mean[\bfW_i \bfW_j \bfW_i]$, which is derived as follows:
\begin{align}
\mean[\bfW_i \bfW_j \bfW_i]
    &=  \mean[\bfW_i \mean[\bfW_j] \bfW_i]
     =  \mean[\bfW_i \bfSigma_j \bfW_i] \notag \\
    &=  \bfSigma_j^{-\sfrac{1}{2}} 
        \mean[ \bfSigma_j^{\sfrac{1}{2}} \bfW_i \bfSigma_j \bfW_i \bfSigma_j^{\sfrac{1}{2}} ] 
        \bfSigma_j^{-\sfrac{1}{2}} \notag \\
    &=  \bfSigma_j^{-\sfrac{1}{2}} 
    \left[ \trace{\bfSigma_i \bfSigma_j} \bfSigma_i \bfSigma_j + \bfSigma_i^2 \bfSigma_j^2 \right]
    \bfSigma_j^{-\sfrac{1}{2}} \label{eq:bfWi_bfWj_bfWi} \\
    &=  \trace{\bfSigma_i \bfSigma_j} \bfSigma_i + \bfSigma_i^2 \bfSigma_j,
\end{align}
where Eq. \eqref{eq:bfWi_bfWj_bfWi} is due to the fact that $\mean[ \bfSigma_j^{\sfrac{1}{2}} \bfW_i \bfSigma_j \bfW_i \bfSigma_j^{\sfrac{1}{2}} ]$ computes the second moment of
    $\bfSigma_j^{\sfrac{1}{2}} \bfW_i \bfSigma_j^{\sfrac{1}{2}}$, which follows complex Wishart distribution with unit degree of freedom and a diagonal covariance $\bfSigma_j^{\sfrac{1}{2}} \bfSigma_i \bfSigma_j^{\sfrac{1}{2}}$. 

Then, $\mathbf{\Pi}(t)$ is factored out and the resulting equation now contains $\mathbf{\Pi}(t)^{-1}$. 
    Since the cumulants are evaluated as ${t \!\to 0}$, both
    ${\mathbf{\Pi}(t) \!= \bfI_T + O(t)}$ and the Neumann series expansion 
    ${\mathbf{\Pi}(t)^{-1} \!= \bfI_T + O(t)}$ can be used to simplify the resulting equation, yielding the following:
\begin{align}
\mathbf{\Psi}_{\bfW}(t)
    =   \sum_{r=1}^{R} \mathbf{\Psi}_{\bfW_r}(t) + \log\left[ \bfI_T + \frac{(\mathbbm{i} t)^3}{3!} \bfTheta_3 \right] + O(t^4).
\end{align}

Taking derivatives both side and evaluating as $t\to 0$ yields Eq. \eqref{eq:mat_cumulant_W}, thus completing the proof of Theorem \ref{theo:mat_cumulant_W}.
% =================================================================================================================================================

\section{Proof of Corollary \ref{lem:cumul_sum_of_iso_wishart}}
\label{apx:cumul_sum_of_iso_wishart}
The sum of statistically independent isotropic complex Wishart matrices can be rewritten in the following stochastic representation:
\begin{align}
\bfX \stackrel{\mathbbm{P}}{=} \sum_{i=1}^{p}\sum_{j=1}^{\nu_i} {\bfX_{i, j}} ,
\end{align}
where $\{ \bfX_{i, j} \}_{j=1}^{\nu_i}$ are statistically i.i.d. complex Wishart matrices with covariance matrix $\sigma_i \bfI_q$. 
    Based on Eq. \eqref{eq:exact_moment_Wk} and Eq. \eqref{eq:ck_trace}, the $n^{\uth}$ moment of $\bfX_{i, j}$ is derived as shown below:
\begin{align}
\mean[\bfX_{i, j}^n]
    &=   \tsum_{j=1}^{n}{ \frac{(n-1)!}{(n-j)!} \smean\left[ \trace{\bfX_{i,j}}^{n-j} \right] {\sigma_i^{j}} \bfI_q } \notag\\
    &=   {\sigma_i^{n}} \tsum_{j=1}^{n}{ \frac{(n-1)!}{(n-j)!} \frac{ (q+n-j-1)! }{ (q-1)! } \bfI_q } \\
    &=   \frac{ (q+n-1)! }{ q! } {\sigma_i^{n}} \bfI_q, \quad n \ge 1.
\end{align}

Since $\trace{\bfX_{i,j}}$ is Gamma-distributed with shape $q$ and scale $\sigma_i^2$, thus $\smean\left[ \trace{\bfX_{i,j}}^{n-j} \right]
    =   \frac{ (q+n-j-1)! }{ (q-1)! } \sigma_i^{n-j}$.
Accordingly, the matrix-valued moment generating function of $\bfX_{i, j}$ is obtained as:
\begin{align}
\mathbf{\Phi}_{i,j}(t)
    &=   \bfI_q + \sum_{n=1}^{ \infty }  \frac{ t^n }{ n! } \frac{ (q+n-1)! }{ q! } \sigma_i^{n} \bfI_q \\
    &=   \left[ 1 + \frac{(1 - \sigma_i t)^{-q} - 1}{q} \right] \bfI_q.
\end{align}

Thus, the matrix-valued cumulant of $\bfX$ is obtained as shown below:
\begin{align}
\mathbf{\Psi}(t)
    =  \sum_{i=1}^{p} \sum_{j=1}^{\nu_i} \log \left[ \mathbf{\Phi}_{i,j}(t) \right] 
        +   \mathbf{\Psi}_{\bfDelta}(t),
\label{eq:mat_cumulant_X}
\end{align}
where $\mathbf{\Psi}_{\mathbf{\Delta}}(t)$ is induced by the 
higher homogeneous Lie polynomial in $\{\bfX_{i,j}\}$.
    We ignore $\mathbf{\Psi}_{\bfDelta}(t)$, since evaluating the infinite series of commutator is mathematically intractable for high-order moments. 
In such a fashion, the eigenvalue moment and cumulant generating functions of $\bfX$ are obtained as follows:
\begin{align}
\phi(t)
    &\propto  \frac{ \trace{ e^{ \mathbf{\Psi}(t) } } }{ q } 
     =  e^{ - \sum\limits_{i=1}^{p} \nu_i \log\left[ q (1 - \sigma_i t )^q  \right] }, \\
\psi(t)
    &\propto  \log\left[ \phi(t) \right]
     =  \sum_{i=1}^{p} \nu_i \log \left[ 1 + \frac{(1 - \sigma_i t)^{-q} - 1}{q} \right],
\label{eq:113}
\end{align}
respectively. 

Let $f(x) \triangleq \log(x)$ and $g_i(t) \triangleq 1 + \frac{(1 - \sigma_i t)^{-q} - 1}{q}$ so that Eq. \eqref{eq:113} is rewritten as
    $\psi(t) = \sum_{i=1}^{p} \nu_i f(g_i(t))$, 
then the $r^{\uth}$ and $t^{\uth}$ derivatives of $f(x)$ and $g(t)$ are given by the following \cite{brychkov2008handbook}:
\begin{align}
f^{(r)}(x) &= \frac{ (-1)^{r-1} (r-1)! }{ x^r }, \\
g^{(t)}(t) &= \frac{ (q+t-1)! }{ q! } (1 - \sigma_i t)^{-q-t} \sigma_i^{m}.
\end{align}

The $n^{\uth}$ eigenvalue cumulant of $\bfX$ is the $n^{\uth}$ order derivatives of $\psi(t) \propto \sum_{i=1}^{p} \nu_i f(g_i(x))$ evaluated at $t \to 0$.
    The higher other derivatives of composite function $f(g_i(x))$ is obtained by using the Fa{\`a} di Bruno's formula \cite{johnson2002curious}, which yields:
\begin{align}
\skmul_{n}[\uplambda_{\bfX}]
    &\propto  \bigg[ \tsum_{i=1}^{p} \nu_i \tsum_{r=1}^{n} f^{(r)}(g_{i}(t))   
    \notag\\
    &\qquad\times
    \left. 
        \mathbbm{B}_{n,r} \! \left( g_{i}^{(1)}(t), g_{i}^{(2)}(t), \dots, g_{i}^{(n-r+1)}(t) \right) \bigg]
    \right|_{t\to 0} \notag \\
    &\propto  \left[ \tsum_{i=1}^{p} \nu_i \sigma_{i}^{n} \right]
    \bigg[ \tsum_{r=1}^{n} \frac{ (-1)^{r-1} (r-1)! }{ q^{r} }
    \notag\\
    &\qquad\qquad\times
    \mathbbm{B}_{n,r} \! \left( (q)_{1}, (q)_{2}, \dots, (q)_{n-r+1} \right) \bigg], \label{eq:98}
\end{align}
where $\mathbbm{B}_{n,r}(\cdot)$ denotes the incomplete Bell polynomial. 
    Note that we used \cite[Eq. (2)]{jin2022partial}, i.e., $\mathbbm{B}_{n,r}( \alpha \beta x_1, \dots, \alpha \beta^{n-r+1} x_{n-r+1} ) 
    = \alpha^r \beta^n \mathbbm{B}_{n,r}( x_1, \dots, x_{n-r+1} )$ to obtain the last equation.
Next, from \cite[Eq. (30)]{jin2022partial}, we deduce that:
\begin{align}
\mathbbm{B}_{n,r} \left( (q)_{1}, \dots, (q)_{n-r+1} \right)
    \!=\!  \frac{ (-1)^{r} }{ r! } \! \tsum_{l=1}^{r} \binom{r}{l} (-1)^l (q l)_{n}.
\end{align}

Plugging the above identity into Eq. \eqref{eq:98} and replacing the component inside the square brackets by $\varRho_{n}[\uplambda_{\bfX}]$, we obtain Eq. \eqref{eq:cumul_sum_of_iso_wishart_1}. This completes the proof of Corollary \ref{lem:cumul_sum_of_iso_wishart}.

\bibliographystyle{IEEEtran}
\bibliography{ref}

\begin{thebibliography}{10}
\providecommand{\url}[1]{#1}
\csname url@samestyle\endcsname
\providecommand{\newblock}{\relax}
\providecommand{\bibinfo}[2]{#2}
\providecommand{\BIBentrySTDinterwordspacing}{\spaceskip=0pt\relax}
\providecommand{\BIBentryALTinterwordstretchfactor}{4}
\providecommand{\BIBentryALTinterwordspacing}{\spaceskip=\fontdimen2\font plus
\BIBentryALTinterwordstretchfactor\fontdimen3\font minus
  \fontdimen4\font\relax}
\providecommand{\BIBforeignlanguage}[2]{{%
\expandafter\ifx\csname l@#1\endcsname\relax
\typeout{** WARNING: IEEEtran.bst: No hyphenation pattern has been}%
\typeout{** loaded for the language `#1'. Using the pattern for}%
\typeout{** the default language instead.}%
\else
\language=\csname l@#1\endcsname
\fi
#2}}
\providecommand{\BIBdecl}{\relax}
\BIBdecl

\bibitem{Zhang2018CM}
J.~Zhang \emph{et~al.}, ``3{D} {MIMO} for {5G} {NR}: Several observations from
  32 to massive 256 antennas based on channel measurement,'' \emph{IEEE Commun.
  Mag.}, vol.~56, no.~3, pp. 62--70, March 2018.

\bibitem{bjornson2019massive}
E.~Bj{\"o}rnson, L.~Sanguinetti, H.~Wymeersch, J.~Hoydis, and T.~L. Marzetta,
  ``Massive {MIMO} is a reality—what is next?: Five promising research
  directions for antenna arrays,'' \emph{Digital Signal Processing}, vol.~94,
  pp. 3--20, 2019.

\bibitem{rappaport2019wireless}
T.~S. Rappaport \emph{et~al.}, ``Wireless communications and applications above
  100 {GH}z: Opportunities and challenges for {6G} and beyond,'' \emph{IEEE
  Access}, vol.~7, pp. 78\,729--78\,757, 2019.

\bibitem{wang2023road}
C.-X. Wang \emph{et~al.}, ``On the road to {6G}: Visions, requirements, key
  technologies, and testbeds,'' \emph{IEEE Commun. Surveys Tut.}, vol.~25,
  no.~2, pp. 905--974, 2023.

\bibitem{ozcelik2005makes}
M.~Ozcelik, N.~Czink, and E.~Bonek, ``What makes a good {MIMO} channel model?''
  in \emph{IEEE VTC}, vol.~1, 2005, pp. 156--160.

\bibitem{wang2023pervasively}
M.~Wang, Y.~He, H.~Wang, C.-X. Wang, and X.~You, ``A pervasively correlated
  channel model for massive {MIMO} transmission,'' \emph{IEEE Trans. Commun.},
  vol.~72, no.~4, pp. 2441--2456, 2023.

\bibitem{demir2024spatial}
{\"O}.~T. Demir, A.~Kosasih, and E.~Bj{\"o}rnson, ``Spatial correlation
  modeling and {RS}-{LS} estimation of near-field channels with uniform planar
  arrays,'' in \emph{2024 IEEE 25th International Workshop on Signal Processing
  Advances in Wireless Communications (SPAWC)}.\hskip 1em plus 0.5em minus
  0.4em\relax Ieee, 2024, pp. 236--240.

\bibitem{psychogios2025dual}
K.~Psychogios, N.~Moraitis, and A.~D. Panagopoulos, ``Dual-polarized {MIMO}
  {UAV}-to-ground tree shadowed channel: Capacity, multiplexing, stationarity,
  and statistical modeling assessment,'' \emph{IEEE Trans. Antennas Propag.},
  vol.~73, no.~6, pp. 3892--3903, June 2025.

\bibitem{sayeed2002deconstructing}
A.~M. Sayeed, ``Deconstructing multiantenna fading channels,'' \emph{IEEE
  Trans. Sig. Process.}, vol.~50, no.~10, pp. 2563--2579, 2002.

\bibitem{kermoal2002stochastic}
J.-P. Kermoal, L.~Schumacher, K.~I. Pedersen, P.~E. Mogensen, and
  F.~Frederiksen, ``A stochastic {MIMO} radio channel model with experimental
  validation,'' \emph{IEEE J. Sel. Areas Commun.}, vol.~20, no.~6, pp.
  1211--1226, 2002.

\bibitem{ying2014kronecker}
D.~Ying, F.~W. Vook, T.~A. Thomas, D.~J. Love, and A.~Ghosh, ``{K}ronecker
  product correlation model and limited feedback codebook design in a {3D}
  channel model,'' in \emph{IEEE ICC}, 2014, pp. 5865--5870.

\bibitem{Raghavan2010TIT}
V.~Raghavan, J.~H. Kotecha, and A.~M. Sayeed, ``Why does the {K}ronecker model
  result in misleading capacity estimates?'' \emph{IEEE Trans. Inf. Theory},
  vol.~56, no.~10, pp. 4843--4864, Oct. 2010.

\bibitem{Weichselberger2006TWC}
W.~Weichselberger \emph{et~al.}, ``A stochastic {MIMO} channel model with joint
  correlation of both link ends,'' \emph{IEEE Trans. Wirel. Commun.}, vol.~5,
  no.~1, pp. 90--100, Jan. 2006.

\bibitem{ren2026multi}
Y.~Ren \emph{et~al.}, ``Multi-state {RIS}-assisted {MIMO} stochastic channel
  modeling and spatial characteristic measurement,'' \emph{IEEE Trans. Wirel.
  Commun.}, vol.~25, pp. 15\,582--15\,596, April 2026.

\bibitem{Zhang2025TIT}
X.~Zhang, S.~Song, and K.~B. Letaief, ``Fundamental limits of non-centered
  non-separable channels and their application in holographic {MIMO}
  communications,'' \emph{IEEE Trans. Inf. Theory}, vol.~71, no.~9, pp.
  6870--6894, Sep. 2025.

\bibitem{gao2009statistical}
X.~Gao, B.~Jiang, X.~Li, A.~B. Gershman, and M.~R. McKay, ``Statistical
  eigenmode transmission over jointly correlated {MIMO} channels,'' \emph{IEEE
  Trans. Inf. Theory}, vol.~55, no.~8, pp. 3735--3750, 2009.

\bibitem{wen2011ergodic}
C.-K. Wen \emph{et~al.}, ``On the {E}rgodic capacity of jointly-correlated
  {R}ician fading {MIMO} channels,'' in \emph{IEEE ICASSP}, 2011, pp.
  3224--3227.

\bibitem{Li2026tutotial}
H.~Li, M.~Nerini, S.~Shen, and B.~Clerckx, ``A tutorial on beyond-diagonal
  reconfigurable intelligent surfaces: Modeling, architectures, system design
  and optimization, and applications,'' \emph{IEEE Commun. Surveys Tut.},
  vol.~28, pp. 4086--4126, Dec. 2026.

\bibitem{wen2011sum}
C.-K. Wen, S.~Jin, and K.-K. Wong, ``On the sum-rate of multiuser {MIMO} uplink
  channels with jointly-correlated {R}ician fading,'' \emph{IEEE Trans.
  Commun.}, vol.~59, no.~10, pp. 2883--2895, 2011.

\bibitem{kollo1995approximating}
T.~Kollo and D.~von Rosen, ``Approximating by the wishart distribution,''
  \emph{Annals of the Institute of Statistical Mathematics}, vol.~47, no.~4,
  pp. 767--783, 1995.

\bibitem{hillier2021moments}
G.~Hillier and R.~Kan, ``Moments of a {W}ishart matrix,'' \emph{Journal of
  Quantitative Economics}, vol.~19, no. Suppl 1, pp. 141--162, 2021.

\bibitem{maiwald2000calculation}
D.~Maiwald and D.~Kraus, ``Calculation of moments of complex {W}ishart and
  complex inverse {W}ishart distributed matrices,'' \emph{IEE Proc. Radar,
  Sonar Navigat.}, vol. 147, no.~4, pp. 162--168, 2000.

\bibitem{Pivaro2017TVT}
G.~F. Pivaro \emph{et~al.}, ``On the exact and approximate eigenvalue
  distribution for sum of {W}ishart matrices,'' \emph{IEEE Trans. Veh.
  Technol.}, vol.~66, no.~11, pp. 10\,537--10\,541, Nov. 2017.

\bibitem{kiessling2004exact}
M.~Kiessling and J.~Speidel, ``Exact ergodic capacity of {MIMO} channels in
  correlated {R}ayleigh fading environments,'' in \emph{Int. Zurich Seminar
  Commun.}, 2004, pp. 128--131.

\bibitem{hachem2008new}
W.~Hachem \emph{et~al.}, ``A new approach for mutual information analysis of
  large dimensional multi-antenna channels,'' \emph{IEEE Trans. Inf. Theory},
  vol.~54, no.~9, pp. 3987--4004, 2008.

\bibitem{bao2015asymptotic}
Z.~Bao, G.~Pan, and W.~Zhou, ``Asymptotic mutual information statistics of
  {MIMO} channels and {CLT} of sample covariance matrices,'' \emph{IEEE Trans.
  Inf. Theory}, vol.~61, no.~6, pp. 3413--3426, 2015.

\bibitem{dumont2010capacity}
J.~Dumont, W.~Hachem, S.~Lasaulce, P.~Loubaton, and J.~Najim, ``On the capacity
  achieving covariance matrix for {R}ician {MIMO} channels: An asymptotic
  approach,'' \emph{IEEE Trans. Inf. Theory}, vol.~56, no.~3, pp. 1048--1069,
  2010.

\bibitem{maaref2007joint}
A.~Maaref and S.~Aissa, ``Joint and marginal eigenvalue distributions of (non)
  central complex {W}ishart matrices and {PDF}-based approach for
  characterizing the capacity statistics of {MIMO} {R}icean and {R}ayleigh
  fading channels,'' \emph{IEEE Trans. Wirel. Commun.}, vol.~6, no.~10, pp.
  3607--3619, 2007.

\bibitem{chen2025divergence}
B.~Chen and J.~Kortje, ``Divergence maximizing linear projection for supervised
  dimension reduction,'' \emph{IEEE Trans. Inf. Theory}, vol.~71, no.~3, pp.
  2104--2115, March 2025.

\bibitem{fevotte2009nonnegative}
C.~F{\'e}votte, N.~Bertin, and J.-L. Durrieu, ``Nonnegative matrix
  factorization with the {I}takura-{S}aito divergence: {W}ith application to
  music analysis,'' \emph{Neural computation}, vol.~21, no.~3, pp. 793--830,
  2009.

\bibitem{golub2009matrices}
G.~H. Golub and G.~Meurant, \emph{Matrices, moments and quadrature with
  applications}.\hskip 1em plus 0.5em minus 0.4em\relax Princeton University
  Press, 2009.

\bibitem{graczyk2003complex}
P.~Graczyk, G.~Letac, and H.~Massam, ``The complex {W}ishart distribution and
  the symmetric group,'' \emph{The Annals of Statistics}, vol.~31, no.~1, pp.
  287--309, 2003.

\bibitem{smith1995recursive}
P.~J. Smith, ``A recursive formulation of the old problem of obtaining moments
  from cumulants and vice versa,'' \emph{The American Statistician}, vol.~49,
  no.~2, pp. 217--218, 1995.

\bibitem{wang2019multivariable}
L.~Wang, Y.~Gao, and N.~Jing, ``On multivariable {Z}assenhaus formula,''
  \emph{Frontiers of Mathematics in China}, vol.~14, no.~2, pp. 421--433, 2019.

\bibitem{horn2012matrix}
R.~A. Horn and C.~R. Johnson, \emph{Matrix analysis}.\hskip 1em plus 0.5em
  minus 0.4em\relax Cambridge university press, 2012.

\bibitem{golberg1972derivative}
M.~A. Golberg, ``The derivative of a determinant,'' \emph{The American
  Mathematical Monthly}, vol.~79, no.~10, pp. 1124--1126, 1972.

\bibitem{brychkov2008handbook}
Y.~A. Brychkov, \emph{Handbook of special functions: derivatives, integrals,
  series and other formulas}.\hskip 1em plus 0.5em minus 0.4em\relax Chapman
  and Hall/CRC, 2008.

\bibitem{johnson2002curious}
W.~P. Johnson, ``The curious history of {F}a{\`a} di {B}runo's formula,''
  \emph{The American mathematical monthly}, vol. 109, no.~3, pp. 217--234,
  2002.

\bibitem{jin2022partial}
S.~Jin, B.-N. Guo, and F.~Qi, ``Partial {B}ell polynomials, falling and rising
  factorials, {S}tirling numbers, and combinatorial identities,'' \emph{CMES
  Comput. Model. Eng. Sci}, vol. 132, no.~3, pp. 781--799, 2022.

\end{thebibliography}

\end{document}